\pdfoutput=1
\documentclass[11pt]{article} 

\usepackage{epsfig} 
\usepackage{amssymb,amsmath,amsthm,amscd}
\usepackage{latexsym} 
\usepackage{cancel}
\usepackage{xcolor}
\usepackage{amssymb}

\usepackage[margin=1.05in]{geometry}
\usepackage{soul}
\newtheorem{thm}{Theorem}[section] 
 
\newtheorem{cor}[thm]{Corollary} 
\newtheorem{lem}[thm]{Lemma}

\theoremstyle{definition}

\theoremstyle{remark} 
\newtheorem{rem}{Remark}[section]  
\def\eqref#1{(\ref{#1})} 
\newcommand{\bbR}{{\mathbb{R}}}

\newcommand {\mat}      [1] {\left[\begin{array}{#1}}
\newcommand {\rix}          {\end{array}\right]}
\newcommand {\de}      [1] {\left|\begin{array}{#1}}
\newcommand {\nt}          {\end{array}\right|}

\newcommand{\bstar}       {\begin{eqnarray*}}
\newcommand{\estar}       {\end{eqnarray*}}
\newcommand{\eqn}       {\begin{eqnarray}}
\newcommand{\enn}       {\end{eqnarray}}
\newcommand{\eq}[1]   {\begin{equation}\label{#1}}
\newcommand{\en}      {\end{equation}}

\begin{document}  
\begin{titlepage}  
\title{Geometric Mean of Concentrations and Reversal Permanent Charge in Zero-Current Ionic Flows via Poisson-Nernst-Planck Models}  
\author{Hamid Mofidi\footnote{Department of Mathematics, University of Iowa, Iowa City, IA 52242, ({\tt hamid-mofidi@uiowa.edu}).}}
\date{\today} 
\end{titlepage}

\maketitle    

\begin{abstract} This work examines the geometric mean of concentrations and its behavior in various situations, as well as the reversal permanent charge problem,  the charge sharing seen in x-ray diffraction. Observations are obtained from analytical results established using geometric singular perturbation analysis of classical Poisson-Nernst-Planck models.
 For ionic mixtures of multiple ion species Mofidi and Liu [{\em SIAM J. Appl. Math. {\bf 80} (2020), 1908-1935}]  centered two ion species with unequal diffusion constants to acquire a system for determining the reversal potential and reversal permanent charge. They studied the reversal potential problem and its dependence on diffusion coefficients, membrane potential, membrane concentrations, etc. Here we use the same approach to study the dual problem of reversal permanent charges and its dependence on other conditions. We consider two ion species with positive and negative charges, say Ca$^+$ and Cl$^-$, to determine the specific conditions under which the permanent charge is unique. Furthermore, we  investigate the behavior of geometric mean of concentrations for various values of transmembrane potential and permanent charge.

 \end{abstract} 


\section{Introduction.} 
\setcounter{equation}{0}
The nervous system is a too complicated and intricate part of an animal that is especially critical for transmitting signals between different body parts. It recognizes environmental changes that affect the body, then
works together with other body systems to respond to detected changes. It contains a large number of neurons that are electrically excitable cells. Neurons send electric signals to each other through thin fibers called axons, which generates chemicals known as neurotransmitters to be delivered at synapses. This electric signal, propagating along the axon, is a variation of polarization of transmembrane electrostatic potential called an action potential. An action potential is induced by the opening and closing many ion channels distributed on the axon membrane. Ion channels,  proteins embedded in membranes,  provide a major pathway for cells to communicate with each other and with the outside to transform signals and to conduct group tasks (\cite{BNVHEG, Eis00, Hil01, Hille89}).
 The essential structure of an ion channel is its shape and its permanent charge.  The shape of a typical channel could be approximated as a cylindrical-like domain.
Within an ion channel, amino acid side chains are distributed mainly over a ``short" and ``narrow" portion of the channel, with acidic side chains contributing negative charges and basic side chains providing positive charges. It is specific of side-chain distributions, which is referred to as the permanent charge of the ion channel.  The function of channel structures is to select the types of ions and to facilitate the diffusion of ions across cell membranes. 
 
The multi-scale feature of the problem with multiple physical parameters enables the system to have high flexibility and to show rich phenomena/behaviors (\cite{BKSA09, Eis}). On the other hand, the same multi-scale characteristic with multiple physical parameters presents a remarkably demanding task for anyone to derive meaningful information from experimental data, also given the fact that the internal dynamics cannot be discerned with the present technique.

To reveal mechanisms of perceived biological phenomena and explore new aspects, the role of mathematical analysis is inevitable. There have been some successes recently in probing Poisson-Nernst-Planck (PNP) models for ionic flows through ion channels \cite{EL07, ELX15, JEL17, JLZ15, Liu05, Liu09, LX15, PJ}. 
Centering specific critical characteristics of the biological systems, PNP models serve suitably for analysis and numerical simulations of ionic flows. One can acquire PNP systems as diminished models from molecular dynamic models, Boltzmann equations, and variational principles \cite{Bar, HEL10, HFEL12, SNE01}. 
%%%%%%%%%%%%%%%%%%%%%%%%%%%%%
There are various types of PNP models:\\ 
\noindent~(i) {\em The classical PNP treats dilute ionic mixtures, where no ion-to-ion interactions are involved.}\\
(ii) {\em The hard-sphere PNP reflects volume exclusive by employing ions as hard-spheres. }\\
%%%%%%%%%%%%%%%%%%%%%%%%%%%%%%%%%%%
%%%%%%%%%%%%%%%%%%%%%%%%%%%%%%%%%%%%
More sophisticated models have also been studied in \cite{BKSA09, CEJS95, EHL10, SR81}, etc. 
It is challenging, though, to achieve analytical or computational results from complicated models.

 In this work, we are engaged in geometric mean of concentrations and reversal permanent charges that are determined by zero total currents.  We study the connection of these quantities with the membrane potentials or diffusion constants. The total current $I=I(V, Q)$ depends on the transmembrane potential $V$ and the permanent charge $Q$. For fixed transmembrane potential $V$, a reversal permanent charge $Q=Q_{rev}(V)$  is a charge that generates zero current $ I(V, Q_{rev}(V))=0$.
We employ the classical PNP model and consider a cylinder-like channel to fulfill the basic understanding of plausible effects of general diffusion coefficients in ionic channels.

To appreciate the significance of permanent charges in ionic channels, we emphasize that permanent charges in ionic channels perform the role of doping profiles in semiconductor devices. Doping gives the charges what acid and base side chains provide in ionic channels.
  Both ionic channels and semiconductor devices employ atomic-scale constructions to regulate macroscopic flows from one pool to another. Holes and electrons are the cations and anions of semiconductors, respectively. Ions usually flow as quasi-particles flow in semiconductors that depends on controlling movement and diffusion of quasi-particles of charge in transistors and integrated circuits.
 Doping is the process of adding impurities into primary semiconductors to strengthen its electrical, optical, and structural features \cite{BEMP01, Rous90, Warn01}.  

The role of diffusion constants, $\mathcal{D}_j$'s is also essential.
 The authors of \cite{ELX15} explored the problem of determining reversal permanent charges for the case when all diffusion constants are equal. However, the identical diffusion constants case is degenerate, known from the biological perspective. 
 The problem with unequal diffusion coefficients has been considered in some works. In \cite{BKSA09}, the authors discussed how mobilities and their spatial inhomogeneities are affected by other parameters.  In \cite{BK18},  the authors carried a perturbation inquiry from a time-independent and spatially homogeneous equilibrium solution. 
Two time scales of the dynamics are classified from the $O(\epsilon)$ terms. In particular, for the equal diffusion coefficient, the authors show that the diffusion process for $O(\varepsilon)$ terms does not occur -- a vital effect of unequal ionic mobilities. 
  In \cite{HBRR18},  the authors reviewed the cases with unequal mobilities by calculations of a wholly nonlinear electrokinetic model. They recognized the appearance of a steady long-range field due to unequal mobilities.

 In this work, we are mainly inspired by the effect of unequal diffusion coefficients and boundary concentrations on the geometric mean of concentrations and reversal permanent charges for the zero-current problem.
In \cite{ML19}, the authors used the geometric singular perturbation (GSP) framework developed in \cite{EL07, Liu05, Liu09} for analyzing PNP models for ionic flow to arrange a system of algebraic equations for the problem. The difference between $D_1$ and $D_2$ makes the system becomes a complex nonlinear algebraic system that is further reduced to two nonlinear equations that appeared to work satisfactorily and allow one to explore the zero-current problem. We utilize the same structure in this paper.

To underline the leading portions of this document (as well as some in \cite{JLZ15, ML19, MEL20}), we first desire to stress that employing the geometric analysis enables one to express and classify quantities and properties that are crucial to biology, and additionally to present the quantitative and qualitative perception and predictions.  In this work, we show a derivation of a mathematical system for the zero-current problem that we use to learn reversal permanent charge and the geometric mean of concentrations and their dependence on the other parameters like membrane potential, boundary concentrations, and diffusion constants.
Some numerical simulations have been provided throughout the text to support the theoretical conclusions and give the interested reader a sharp comprehension of the claims.

Throughout the paper, we obtain the numerical results from the algebraic systems \eqref{G1G2Sys} and  \eqref{G} that have been obtained from the governing system \eqref{Matching}.  The nonlinear algebraic systems are solved by \textsc{Matlab}\textsuperscript{\textregistered}(Version 9.5) function {\em fsolve} that uses the trust-region dogleg algorithm, that is based on the interior-reflective Newton method defined in \cite{CL96}.

\medskip

This paper is constructed as follows. The classical PNP model for ionic flows is recalled in Section \ref{sec-PNP} to set the stage for analyses in the next sections. 
In Section \ref{Sec-GSPonPNP}, we apply the GSP theory on the PNP system, with zero-current, to convert the BVP to a connecting system to end up with a nonlinear algebraic system of equations, called the matching system. 
In Section \ref{sec-ZeroCur}, we reduce the matching system to two nonlinear equations. In particular, we analyze the geometric mean of concentrations and reversal permanent charge. We review some concluding remarks in Section \ref{ConSec}.

%%%%%%%%%%%%%%%%%%%%%%%%%%%%%%%%%%%%%%
%%%%%%%%%%%%%%%%%%%%%%%%%%%%%%%%%%%%%%

\subsection{ PNP Systems for Ion Channels.}\label{sec-PNP}
\setcounter{equation}{0}
The PNP equations have been simulated and computed to a considerable extent \cite{ CE, ChK, BCE, HCE, IR}.  From those simulations, one can see that mathematical boundary conditions, i.e., macroscopic reservoirs, need to be incorporated in the mathematical formulation to describe the actual behavior of channels \cite{GNE, NCE}.
For an ionic mixture of $n$ ion species, the PNP model is, for $k=1,2,..., n$,
\begin{align}\label{PNP}
\begin{split} 
 \text{Poisson:} \quad & \nabla \cdot\Big( \varepsilon_r(\overrightarrow{X}) \varepsilon_0 \nabla \Phi \Big) = -e_0 \Big( \sum_{s=1}^n z_s C_s + \mathcal{Q}(\overrightarrow{X})\Big), \\
 \text{Nernst-Planck:} \quad & \partial_t C_k + \nabla \cdot  \overrightarrow{\mathcal{J}}_k =0, \quad -  \overrightarrow{\mathcal{J}}_k =\dfrac{1}{k_BT} \mathcal{D}_k(\overrightarrow{X}) C_k \nabla \mu_k,
 \end{split}
\end{align} 
where $\overrightarrow{X} \in \Omega$ (which is a three-dimensional cylindrical-like domain)  representing the channel of length $\hat{L} ~\mbox{nm} (= \hat{L} \times 10^{-9}\mbox{m})$,  $\varepsilon_r(\overrightarrow{X})$ is the  relative dielectric coefficient (with unit 1), 
 $\varepsilon_0 \approx 8.854 \times 10^{-12}~\mbox{Fm}^{-1}$ is the vacuum permittivity, $e_0 \approx 1.602\times 10^{-19}\mbox{C}$ (coulomb) is the elementary charge, ${\cal Q}(\overrightarrow{X})$ represent the permanent charge density of the channel (in $\mbox{M} = \mbox{Molar}=\mbox{{mol}/{L}}$ ),  $k_B\approx 1.381 \times 10^{-23} \mbox{JK}^{-1}$  is the Boltzmann constant, $T$ is the absolute temperature ($T\approx 273.16~ \mbox{K}=$kelvin, for water),  and $\Phi$ is the electric potential (with the unit $\text{V}=\text{Volt}=\text{JC}^{-1}$). For  the $k$-th ion species, $C_k$ is the concentration (with unit   $\mbox{M}$), $z_k$ is  the valence (the number of charges per particle with unit $1$), and $\mu_k$ is the electrochemical potential (with unit $\text{J}=\text{CV}$). The flux density $ \overrightarrow{{\cal J}}_k(\overrightarrow{X})$ (with unit $\mbox{mol }\mbox{m}^{-2}\mbox{s}^{-1}$) is the number of particles across each cross-section in per unit time, ${\cal D}_k(\overrightarrow{X})$ is the diffusion coefficient (with unit $\mbox{m}^2/\mbox{s}$), and $n$ is the number of distinct types of ion species (with unit $1$).  
%%%%%%%%%%%%%%%%%%%%%%%%%%%%%%%%%%%%%%%%%%

Since ion channels have thin cross-sections comparative to their lengths, three-dimensional PNP systems can be reduced to quasi-one-dimensional models (\cite{LW10}). The quasi-one-dimensional steady-state PNP model is, for $k = 1, 2, ..., n,$
%%%%%%%%%%%%%%%%%%%%%%%%%%%%%%%%%%%%%%%%
\begin{align}\label{1dPNP}
\begin{split} 
 \frac{1}{ \mathcal{A}(X)}  \frac{d}{dX}\left({\varepsilon}_r(X)\varepsilon_0 \mathcal{A}(X) \frac{d \Phi}{d X}\right)=&-e_0\left( \sum_{s=1}^nz_s C_s + \mathcal{Q}(X)\right), \\ 
\frac{d \mathcal{J}_k}{d X}  =0, \quad  -\mathcal{J}_k=& \dfrac{1}{k_BT}\mathcal{D}_k(X)\mathcal{A}(X)C_k\frac{d \mu_k}{dX},   
\end{split} 
\end{align} 
where  $X$ is the coordinate along the channel, $\mathcal{A}(X)$ is the area of cross-section of the channel over location $X$, and $\mathcal{J}_k$ (with unit $\mbox{mol/}\mbox{s}$)  is the total flux through the cross-section. 
%%%%%%%%%%%%%%%%%%%%%%%%%%%%%%%%%%%%%%%%%%
We apply the following boundary conditions to the system (\ref{1dPNP}), for $k=1,2,\cdots, n$, 

\begin{equation}\label{BV} 
\Phi(0)={\cal V},  \quad C_k(0)=L_k>0; \quad \Phi(\hat{L})=0,\quad C_k(\hat{L})=R_k>0.
\end{equation} 
%%%%%%%%%%%%%%%%%%%%%%%%%%%%%%%%%%%%%%%
%%%%%%%%%%%%%%%%%%%%%%%%%%%%%%%%%%%%%%%%

\noindent One often uses the electroneutrality conditions on the boundary concentrations because the solutions are made from electroneutral solid salts,
 
 \begin{align}\label{neutral}
 \sum_{s=1}^nz_sL_s= \sum_{s=1}^nz_sR_s=0.
 \end{align}
 %%%%%%%%%%%%%%%%%%%%%%%%%%%%%%%%%%%
 
There is a sharp layer for electrical potential and ion concentrations near the encounter between the zero and non-zero permanent charges. This occurs near $x = a$ and $x = b$, in our calculations, where the permanent charge jumps (see how $Q(x)$ is defined in \eqref{Q}). Since our boundary conditions are forced at $x= 0 < a $ and at $ x=1 > b$, one can disregard the changes in the boundary concentrations. Nevertheless, our approach can regulate the boundary layers even when the boundary conditions do not fit the electroneutrality condition \cite{EL07}.

The electrochemical potential ${\mu}_k(X)$ for the $k$-th ion species consists of the ideal component ${\mu}_k^{id}(X)$ and the excess component ${\mu}_k^{ex}(X)$, i.e., ${\mu}_k(X)={\mu}_k^{id}(X)+{\mu}_k^{ex}(X)$. The  excess electrochemical potential ${\mu}_k^{ex}(X)$   accounts for the finite size effect of ions. 
It is needed whenever concentrations exceed, say 50~\text{mM}, as they almost always do in technological and biological situations and often reach concentrations $1\mbox{M}$ or more. 
The classical PNP model only deals with the ideal component ${\mu}_k^{id}(X)$, which disregards ions-size and displays the dilute ions-entropy in water. Dilute solutions tend to approach ideality as they proceed toward infinite dilution; that is,

\begin{equation}\label{Ideal}
{\mu}_k(X)={\mu}_k^{id}(X)= z_k e_0 \Phi(X)+k_B T \ln \frac{C_k(X)}{C_0},
\end{equation}

\noindent
where one may take $C_0=\max_{1\le k\le n}\big\{ L_k, R_k, \sup_{X\in[0,\hat{L}]}|{\cal Q}(X)|\big\}$ as the characteristic concentration of the problems. 

%%%%%%%%%%%%%%%%%%%%%%%%%%%%%%%%%%%%%%%%%%%%

For given $\mathcal{V}$, $\mathcal{Q}(X)$, $L_k$'s and $R_k$'s,   if  $(\Phi(X), C_k(X), \mathcal{J}_k)$ is a solution of the boundary value problem (BVP) of (\ref{1dPNP}) and (\ref{BV}), then  the electric current $\mathcal{I} $ is   
 $
\mathcal{I}= e_0 \sum_{s=1}^nz_s \mathcal{J}_s.
 $
%%%%%%%%%%%%%%%%%%%%%%%%%%%%%%%%%%%%%%%%%%%%%%%
For an analysis of the boundary value problem (BVP) (\ref{1dPNP}) and (\ref{BV}), we work on a dimensionless form. Set 
${\cal D}_0=\max_{1\le k\le n}\{\sup_{X\in [0,\hat{L}]}{\cal D}_k(X)\}\;\mbox{ and }\; \bar{\varepsilon}_r=\sup_{X\in [0,\hat{L}]} \varepsilon_r(X)$. Then let

\begin{equation}\label{DimToDimless}
\begin{aligned}
&\varepsilon^2=\frac{\bar{\varepsilon}_r\varepsilon_0k_BT}{e_0^2\hat{L}^2C_0},\quad
  \hat{\varepsilon}_r(x)=\frac{\varepsilon_r(X)}{\bar{\varepsilon}_r},\quad x=\frac{X}{\hat{L}},\quad  h(x)=\frac{\mathcal{A}(X)}{\hat{L}^2},\quad D_k(x)=\frac{{\cal D}_k(X)}{{\cal D}_0},\\
  & Q(x)=\frac{{\cal Q}(X)}{C_0},  \quad \phi(x)=\frac{e_0}{k_BT}\Phi(X), \quad c_k(x)=\frac{C_k(X)}{C_0},\quad \hat{\mu}_k=\frac{1}{k_BT}\mu_k, \quad J_k=\frac{{\cal J}_k}{ \hat{L} C_0{\cal D}_0}. 
 \end{aligned}
\end{equation}
%%%%%%%%%%%%%%%%%%%%%%%%%%%%%%%

\noindent
In terms of the new variables, the  BVP    (\ref{1dPNP}) and (\ref{BV}) become, for $k=1,2,\cdots,n$,

\begin{align}\label{1dPNPdim}
\begin{split} 
 \frac{\varepsilon^2}{ h(x)}  \frac{d}{dx}\left(\hat{\varepsilon}_r(x)h (x)\frac{d}{dx}\phi\right)=&- \sum_{s=1}^nz_s c_s -  Q(x), \\ 
\frac{d J_k}{dx}  =0, \quad  -J_k=&h (x)D_k(x)c_k\frac{d  }{dx}\hat{\mu}_k ,   
\end{split} 
\end{align} 

\noindent
with the boundary conditions  

\begin{equation}\label{1dBV} 
\phi(0)=V=\frac{e_0}{k_BT}{\cal V},  \quad c_k(0)=l_k=\frac{L_k}{C_0}; \quad \phi(1)=0,\quad c_k(1)=r_k=\frac{R_k}{C_0}.
\end{equation}

%%%%%%%%%%%%%%%%%%%%%%%%%%%%%%%%%%%%%%%%%%%%%
\begin{rem} It is reasonable to assume that $\varepsilon>0$ in system \eqref{1dPNPdim} is small because if $ \hat{L}= 2.5 ~\mbox{nm}$ and $C_0 = 10 ~\mbox{M}$, then we obtain $\varepsilon \approx 10^{-3}$ \cite{EL17}. 
The smallness of $\varepsilon$ will later let us treat the system \eqref{slow} of the dimensionless problem as a singularly perturbed problem that can be analyzed by the GSP theory. As we will discuss in more detail later in Section \ref{Sec-GSPonPNP}, the GSP Theory employs the modern  invariant manifold theory from nonlinear dynamical system theory to examine the entire structure, i.e., the phase space portrait of the dynamical system, and should not be confused with the classical singular perturbation theory that uses, for example,  matched asymptotic expansions.
\qed
\end{rem}
%%%%%%%%%%%%%%%%%%%%%%%%%%%%%%%%%%%
As seen in Fig \ref{Fig-hx}, for any point $x\in [0,1]$, $h(x)$ (the right panel) is the cross-section area of the channel in a dimensionless form corresponding to $X\in [0,\hat{L}]$ in $\Omega$, (the left panel) which is the three-dimensional form of the channel.
%%%%%%%%%%%%%%%%%%%%%%%%%%%%%%%%%%%%%%%%%
 \begin{figure}[h]\label{Fig-hx}
	\centerline{\epsfxsize=5.5in \epsfbox{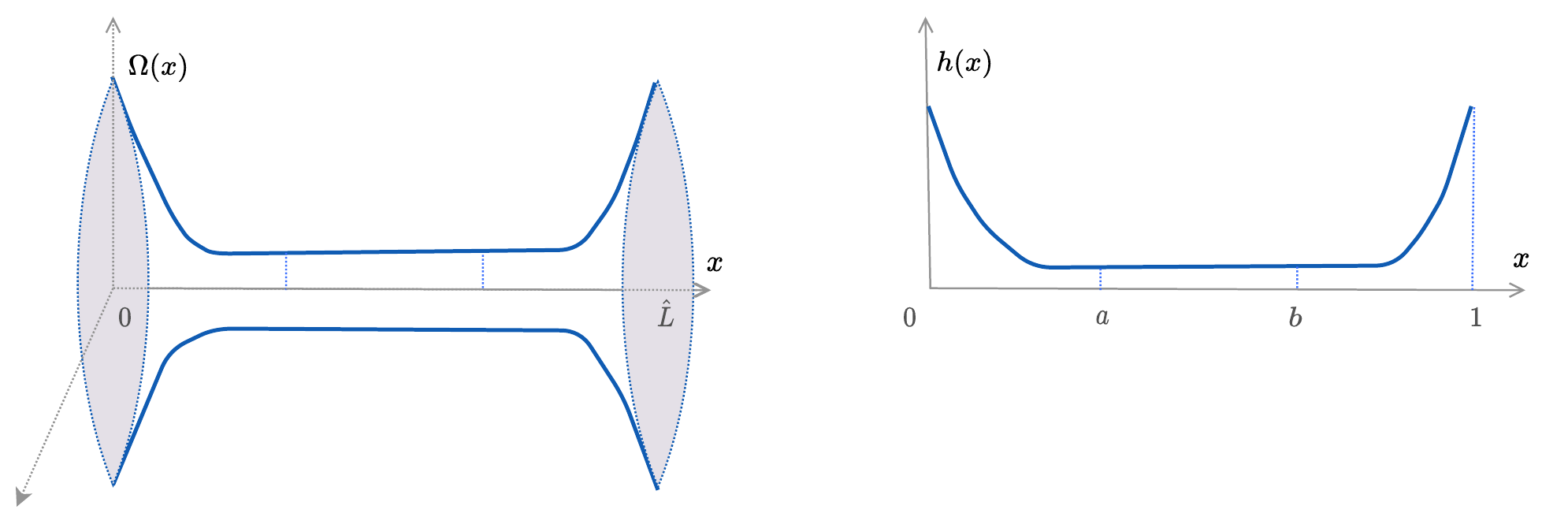} }\vspace*{-.1in}
	\caption{\em  The left panel shows $\Omega(x)$, which is a three-dimensional cylindrical like domain representing a channel of length $\hat{L}$; the right panel shows $h(x)$ which is the dimensionless form of the area of cross-section of the channel.  }
	\label{Fig-hx}
\end{figure}
%%%%%%%%%%%%%%%%%%%%%%%%%%%%%%%%%%%%%%%%%%%
\begin{rem}\label{rem-dim1}
 At this step, we would like to emphasize some tips:

(i)~ We first recall that dimensionless variables are suitable for explaining mathematical and general physical relations, and dimensional quantities are required to reveal how evolution has utilized those relations. We frequently switch from dimensional form to the dimensionless one and conversely throughout the script. 
%%%%%%%%%%%%%%%%%%%%%%%%%%%%%%%%%%%%
The quantities $V,l,r, c_k,Q,  D_k$, and $J_k$ are dimensionless values corresponding to the dimensional quantities  $\mathcal{V},L,R, C_k,\mathcal{Q},  \mathcal{D}_k$, and $\mathcal{J}_k$, respectively, acquired from \eqref{DimToDimless}.  

(ii) The genuine dimensional forms of quantities have been applied for all figures throughout the paper, that is, 

$$
\begin{aligned}
C_k =& C_0 c_k~ (\mbox{M}),\qquad  \mathcal{Q}= C_0 Q~ (M), \qquad  
 \mathcal{J}_k = \hat{L} C_0 \mathcal{D}_0 J_k  ~(\mbox{mol}/{\mbox{s}}),
\end{aligned}
$$
 and we take
$C_0 = 10 \mbox{M}$,  $\hat{L} = 2.5\mbox{ nm}$ and $
\mathcal{D}_0 =2.032\times 10^{-9} ~ \mbox{m}^2/\mbox{s}.
$
Also, for diffusion constants $\mathcal{D}_k$ we consider combinations of the following \cite{LE14, Van93},  
\begin{equation}\label{Dif-Coef}
\begin{aligned}
& 1.334 \times 10^{-5} ~ \mbox{cm}^2/\mbox{s}~ \text{for Na}^+,\hspace*{.06in} \qquad 
 2.032 \times 10^{-5} ~ cm^2/s~ \text{for Cl}^-,\\
& 0.792 \times 10^{-5} ~cm^2/s~ \text{for Ca}^{2+},\qquad 0.923 \times 10^{-5} ~ \mbox{cm}^2/\mbox{s}~ \text{for CO}_3^-,\\
\end{aligned}
\end{equation}
where $Na^+, Cl^-, Ca^{2+},$ and $CO_3^{2-}$ respectively stand for Sodium, Chloride, Calcium and Carbonate. \qed
\end{rem}

%%%%%%%%%%%%%%%%%%%%%%%%%%%%%%%%%%%%%%%%%%%%%%
%\subsection{Setup of the Problem for Application of GSP Theory.}\label{setups}
  We now designate the case we will study in this paper. We will investigate  a simple  setup, the classical PNP model \eqref{1dPNPdim} with the boundary conditions \eqref{1dBV},  and ideal electrochemical potential \eqref{Ideal}.
More clearly, we assume
{\em \begin{itemize}
 \item[(A0)]  The ionic mixture consists of two ion species  with valences $z_1=-z_2=1$;
  \item[(A1)] $D_k(x)=D_k$  for $k=1,2$ is a constant and  $\hat{\varepsilon}(x)=1$;
\item[(A2)] Electroneutrality boundary conditions (\ref{neutral}) hold;
\item[(A3)] The permanent charge $Q$ is piecewise constant defined as, 
 \begin{align}\label{Q}
 Q(x)=\left\{\begin{array}{ll}
 Q_1=Q_3=0, & x\in (0,a)\cup (b,1),\\
 Q_2, & x\in (a,b),
 \end{array}\right.
 \end{align}   
 where $Q_2$ is a constant.
 \end{itemize}}

%%%%%%%%%%%%%%%%%%%%%%%%%%%%%%%%%%%%%%%%%%%%%
%%%%%%%%%%%%%%%%%%%%%%%%%%%%%%%%%%%%%%%%%%%%%
  \section{Application of Geometric Singular Perturbation Theory to  the PNP problem.}\label{Sec-GSPonPNP}
  \setcounter{equation}{0}
A common struggle for nonlinear differential equations is that it is usually impractical to develop specific solutions since they are very diverse, and examining solutions depends on the problems. Time scale separation is a phenomenon that appears in many physical systems like chemical reactions, particles in fluids, etc. In many applications, systems are naturally formulated in fast-slow systems, that are typically high-dimensional systems of nonlinear ordinary differential equations. The GSP Theory is a robust device to analyze multiple scale problems. It relies on advances of invariant manifold theory from nonlinear dynamical system theory and is an alternative and complementary theory to the classical matched asymptotic theory for singular perturbations.

\subsection{Singularly Perturbed Problems: Fast-Slow Systems.}\label{Sec-fast-slow}
One may presume that singularly perturbed problems are more complex than regularly perturbed problems. However, there are significant benefits of singularly perturbed problems over the regularly perturbed problems. One is the approximate decoupling of the full singularly perturbed problems. In general, IVP of singularly perturbed systems is exponentially (in $\varepsilon$) sensitive in initial conditions, but BVP is better behaved.
 There is no globally admitted classification to group all perturbation problems into categories of regular and singular. Nonetheless, some formal and informal tactics are listed below (\cite{Kuehn15}).
 
A differential equation problem involving a small parameter $0 < \varepsilon \ll 1 $ can be called a singular perturbation problem under one of the following definitions:\vspace*{-.1in}
\begin{itemize}
\item[1.] The asymptotic series is not a power series in $\varepsilon$, or if it is, the power series has a vanishing radius of convergence.\vspace*{-.12in}

\item[2.] The solution does not converge uniformly as $\varepsilon \to 0$ to a singular solution for $\varepsilon = 0$.\vspace*{-.12in}

\item[3.] Substituting a power series expansion in $\varepsilon$ yields problems of a ``different type'' from that of the original differential equation.\vspace*{-.12in}

\item[4.] Substituting a (regular) power series expansion in $\varepsilon$ fails, i.e., a problem is singular if it is not regular.
\end{itemize}
A fast-slow vector field (or $(m, n)-$ fast-slow system) is a system of ordinary differential equations taking the form
\begin{equation}\label{fs1}
\begin{aligned}
 \frac{dx}{d\tau}&=~ \overset{.}{x}= f(x,y, \varepsilon),\\
\varepsilon\frac{dy}{d\tau}&=  \varepsilon\overset{.}{y}= g(x,y, \varepsilon),
\end{aligned}
\end{equation}
where $f : \mathbb{R}^m \times \mathbb{R}^n \times \mathbb{R} \to \mathbb{R}^m$, $g : \mathbb{R}^m \times \mathbb{R}^n \times \mathbb{R} \to \mathbb{R}^n$, and $0 < \varepsilon \ll 1$. 
%%%%%%%%%%%%%%%%%%%%%%%%%%%%%%%%%%%
 Furthermore, the $x$ variables are called {\em slow variables}, and the $y$ variables are called {\em fast variables}. Setting $t =\frac{\tau}{\varepsilon}$ gives the equivalent form 
\begin{equation}\label{fs2}
\begin{aligned}
\frac{dx}{dt}&= {x'}= \varepsilon f(x,y, \varepsilon),\\
\frac{dy}{dt}&=  {y'}= g(x,y, \varepsilon).
\end{aligned}
\end{equation}
We refer to $t$ as the {fast time scale} and to $\tau$ as the {slow time scale}.

  The differential-algebraic equation obtained by setting $\varepsilon = 0$
in the formulation of the slow time scale \eqref{fs1} is called the {slow subsystem} or {slow vector field} or {limiting slow system}:
\begin{equation}\label{s0}
\begin{aligned}
\overset{.}{x}&= f(x,y, 0),\quad 0= g(x,y, 0).
\end{aligned}
\end{equation}
The flow generated by \eqref{s0} is called the {slow flow}.
Furthermore, the parameterized system of ODEs obtained by setting $\varepsilon = 0$
in the formulation of the fast time scale \eqref{fs2} is called the {fast subsystem} or {fast vector field} or {limiting fast system}:
\begin{equation}\label{f0}
\begin{aligned}
x'&= 0,\quad y'= g(x,y, 0).
\end{aligned}
\end{equation}
The flow generated by \eqref{f0} is called the {fast flow}.
 The limiting slow system is also referred to as the {reduced problem} and its flow as the {reduced flow}. The limiting fast system is also referred to as the {layer equations} or the {layer problem}.
We call the set
\begin{equation}\label{Z0-SManif}
\mathcal{Z}_0= \big\{(x,y) \in \mathbb{R}^m \times \mathbb{R}^n : g(x,y,0) = 0 \big\},
\end{equation}
the {critical set}, or {slow manifold} or {critical manifold}  if $\mathcal{Z}_0$ is a submanifold of $\mathbb{R}^m \times \mathbb{R}^n$.

 The relation between equilibrium points of the fast flow and the critical manifold $\mathcal{Z}_0$ is particularly simple. 
Any orbit (or a portion of it) of the limiting slow system is called a singular slow orbit.
Similarly, any orbit (or a portion of it) of the limiting fast system is called a singular fast orbit. By a singular orbit of $\varepsilon > 0$ system, we mean that it is a continuous curve in the phase space that is a union of singular slow orbits and singular fast orbit.

In summary, for $\varepsilon > 0$, the slow system \eqref{fs1} and fast system \eqref{fs2} are equivalent;
their limiting versions \eqref{s0} and \eqref{f0} are totally different but are complement to each other
with lower dimensions. In addition to being lower dimensions, as mentioned above, the limiting subsystems
often correspond to ideal physical conditions that are easier to analyze in principle. Therefore, the main objective of the GSP theory is what information from the limiting systems can be lifted to the full system, particularly how limiting fast dynamics and limiting slow dynamics are interplayed.

 \subsection{Geometric Singular Perturbation Theory.}\label{Sec-GSP1}
 In applying the nonlinear dynamical system theory, particularly the invariant manifold theory, to the study of singularly perturbed problems, one tries to understand essential structures of the phase portrait. A general systematic approach goes as follows.
We recall that the slow manifold $\mathcal{Z}_0$ is the set of equilibria of the limiting fast system \eqref{f0}. The linearization at each point $p$ on the slow manifold $\mathcal{Z}_0$ is

\begin{equation}\label{NHmatrix}
\left(\begin{aligned}
&\qquad 0\qquad\qquad 0\\
&D_xg(p,0)\quad D_yg(p,0)
\end{aligned}\right).
\end{equation}
Therefore, $\lambda_0 = 0$ is always an eigenvalue of multiplicity $n$. We call $\lambda_0$ the trivial eigenvalue and all other $m$ eigenvalues of \eqref{NHmatrix} the nontrivial eigenvalues.
Assume that the number of nontrivial eigenvalues in the right complex half-plane, on the imaginary axis, and in the left complex half-plane are $m^u$, $m^c$, and $m^s$ respectively. We denote the corresponding unstable, center, and stable eigenspaces by $E^u$, $E^c$, and $E^s$ with
$$
\dim E^u =m^u,\quad \dim E^c =m^c+n,\quad \dim E^s =m^s.$$
The set $\mathcal{Z}_0$ is called {\em normally hyperbolic} if the $n\times n$ matrix $(D_xg)(p,0)$ of first partial derivatives with respect to the fast variables has no eigenvalues with zero real part for all $p \in \mathcal{Z}_0$.

%\begin{prop}  A subset $S\subset C_0$ is normally hyperbolic if and only if for each $p=(x^*, y^*) \in S$, we have that $y^*$ is a hyperbolic equilibrim point of $y'= g(x^*,y, 0)$.
%\end{prop}

% A normally hyperbolic subset $S \subset C_0$ is called attracting if all eigenvalues of $(D_xg)(p,0)$ have negative real part for $p\in S$; similarly, $S$ is called repelling if all eigenvalues have positive real part. If S is normally
%hyperbolic and neither attracting nor repelling, it is of saddle type.

\begin{thm}({Fenichel's First and Second Theorems}) (\cite{ Hek, Jones95})\\
Suppose $S_0$ is a compact normally hyperbolic  submanifold (possibly with boundary) of the critical manifold $\mathcal{Z}_0$ of \eqref{fs1} and that $f,g \in C^r ~(r <\infty)$, that is they are smooth. Then for $\varepsilon > 0$ sufficiently small:\\
(1) ~Fenichel's first theorem: There exists a manifold $S_{\varepsilon}$, $O(\varepsilon)$ close and diffeomorphic to $S_0$,  that is locally invariant under the flow of the full problem \eqref{fs1}.\\
(2) ~Fenichel's second theorem:  There exist manifolds $W^s(S_\varepsilon)$ and $W^u(S_\varepsilon)$, that are $O(\varepsilon)$  close and diffeomorphic to $W^s(S_0)$ and $W^u(S_0)$, respectively, and that are locally invariant under the flow of the full problem \eqref{fs1}.
\end{thm}

Consider an $(m,n)-$fast-slow system \eqref{fs1}
 and suppose the critical manifold $\mathcal{Z}_0$ in \eqref{Z0-SManif} is normally hyperbolic. Since $\mathcal{Z}_0$ is normally hyperbolic, it follows that $D_y g|_{\mathcal{Z}_0}$ is invertible. Hence, the implicit function theorem locally admits that $\mathcal{Z}_0$ in \eqref{Z0-SManif} becomes $\mathcal{Z}_0 = \big\{(x,h(x)) \in \mathbb{R}^{m+n} \big\}$, which is a graph, where $h : \mathbb{R}^n \to \mathbb{R}^m$ is a map so that $g(x,h(x), 0) = 0$.
To state the Exchange Lemma, we need the following terminology. Two smooth (sub)manifolds $M_1$ and $M_2$ of $\mathbb{R}^k$ intersect transversally at a point $q$ if $T_qM_1$ and $T_qM_2$ together spans $ T_q\mathbb{R}^k = T_qM_1 + T_qM_2.
$
%%%%%%%%%%%%%%%%%%%%%%%%%%%%%%%%%%%
In this case, the intersection $N=M_1\cap M_2$ is a smooth submanifold, and
$
\dim N = \dim M_1 + \dim M_2 - \dim \mathbb{R}^k.
$
Furthermore, if $M_1$ is a submanifold of $U$, then $U$ and $M_2$ intersect transversally too.
Also, $C^1$ perturbations of $M_1$ and $M_2$ still intersect transversally. 
%%%%%%%%%%%%%%%%%%%%%%%%%%%%%%%%%%%%%%

Jones and Kopell (\cite{Jones95, JK94}), with deep insight, extracted an extremely useful consequence, collectedly called Exchange Lemma. They viewed the normally hyperbolic slow manifold together with its invariantly foliated stable and unstable manifolds as a device and tested what the device does to an invariant manifold as the invariant manifold enters a neighborhood of the slow manifold. 
Recall the setup of the singularly perturbed problem in the standard form, i.e. equations \ref{fs1}---\ref{f0}. Moreover, assume that the slow manifold $\mathcal{Z}_0$ is normally hyperbolic, and assume that $k$ eigenvalues of the $n\times n$ matrix $g_y(x, h(x), 0)$ have positive real part and
$l$ eigenvalues of $g_y(x, h(x), 0)$ have negative real part with $k + l = n $. Let $U$ be a neighborhood of $\mathcal{Z}_0$ within which the invariant foliation structures holds. Let $M_\varepsilon$ be an invariant manifold of \eqref{fs1} and \eqref{fs2} for $\varepsilon >0$ and $\dim M_\varepsilon = k +\sigma$ with $ 1 \leq \sigma \leq m$. The limiting manifold
$M_0$ is invariant under the limiting fast system \eqref{f0}. Furthermore, we assume\\
(i). $M_0$ intersects $W^s(\mathcal{Z}_0) \cup U$ transversally. Let $N_0 =M_0 \cap W^s(\mathcal{Z}_0)$, then
$$
\dim N_0 = \dim M_0 + \dim W^s(\mathcal{Z}_0) - \dim(\mathbb{R}^{m+n}) = (k +\sigma) + (m+l) - (m+n) = \sigma.
$$
(ii). The $\omega-$limit set $\omega (N_0)$ of the limiting fast system is a $(\sigma -1)-$dimensional submanifold of $\mathcal{Z}_0$. One should note that $\omega (N_0)$ is nonempty because $\sigma - 1 \geq 0$.\\
(iii). The limiting slow flow is not tangent to $\omega (N_0)$. We remark that $\omega (N_0)$ is not open in $\mathcal{Z}_0$ because $\sigma -1 < m$.\\
The latter statement indicates that under the limiting slow flow denoted by dot, $\omega (N_0)\cdot(0,\tau)$
for any $\tau > 0$ is a submanifold of $\mathcal{Z}_0$ of dimension $\sigma = \dim \omega (N_0) + 1$.

\begin{thm} ({Exchange Lemma for normally hyperbolic slow manifold})(\cite{Jones95})\\
Assume (A1), (A2) and (A3) hold. Let $\tau_1 > \tau_0 >0$ be given. Then, for $\varepsilon > 0$ small, a portion of $M_\varepsilon$ is $C^1~O(\varepsilon)-$close to $W^u\big(\omega (N_0)\cdot(\tau_0,\tau_1)\big)\cap U$.
Note that, $$\dim \big(\omega (N_0)\cdot(\tau_0,\tau_1)\big)= \dim \omega(N_0)+1 = \sigma,$$ and hence,
$$
W^u\big(\omega (N_0)\cdot(\tau_0,\tau_1)\big) = \dim \omega(N_0)+k = \sigma + k = \dim M_\varepsilon.
$$
We remark that all conditions are required in limiting slow and fast systems, but the conclusion is for $\varepsilon > 0$ small.
\end{thm}

The conclusion on $C^1$ closeness stresses that not only the manifolds are $O(\varepsilon)-$close in
the $C^0$ sense, but also their tangent spaces are $O(\varepsilon)-$close to each other.

Standard utilization of the Exchange Lemma is connecting orbits (BVP), heteroclinic and homoclinic orbits. In connecting orbits, $M_\varepsilon$ could be the flow of the set $B_L$ defining the boundary condition on the left. One then requires to discern whether $M_\varepsilon$ intersect the other set $B_R$ defining the boundary condition on the right.
That is what occurs to the PNP system, as we will address in the next section.

  \subsection{Application of GSP Theory to  the BVP (\ref{1dPNPdim}) and (\ref{1dBV}).}\label{Sec-GSPBVP}
  
We rewrite the classical PNP system \eqref{1dPNPdim} into a standard form of singularly perturbed systems and turn the boundary value problem to a connecting problem.
For more details one can read the paper \cite{EL07}.  
Denote the derivative with respect to $x$ by  overdot  and introduce $u=\varepsilon \dot \phi$.
  System~\eqref{1dPNPdim} becomes,   for $k=1,2,$ 
 \begin{align}\label{slow}\begin{split} 
  \varepsilon\dot \phi=&   u,\quad
  \varepsilon \dot u=  -\sum_{s=1}^{ 2}z_sc_s-Q(x) - \varepsilon \frac{h_x(x)}{h(x)} u, \\
  \varepsilon \dot c_k=&-z_kc_ku - \varepsilon  \frac{J_k}{D_k h(x)},  \quad \dot J_k= 0.
 \end{split}
 \end{align}
 
  System~\eqref{slow} will be treated as a dynamical system with the phase space
  $\bbR^{7}$ and the independent variable $x$ is viewed as time for the dynamical system.

%%%%%%%%%%%%%%%%%%%%%%%%%%%%%%%%%%%%%%%%%%%%%%%%%%
%%%%%%%%%%%%%%%%%%%%%%%%%%%%%%%%%%%%%%%%%%%%%%%%%

A  GSP structure to examine the BVP of the classical PNP systems was formed first in \cite{EL07, Liu05} for ionic compounds with two ion species. The model of ion channel properties involves coupled nonlinear differential equations. The GSP theory enables one to make a conclusion about the BVP for $\varepsilon>0$ small from the data of $\varepsilon=0$ limit systems. Another unique structure is that a state-dependent scaling of the independent variable transforms the nonlinear limit slow system to a linear system with constant coefficients. The coefficients depend on unknown fluxes to be determined as part of the whole problem, which is mathematical evidence for its strong dynamics. Consequently, the existence, multiplicity, and spatial profiles of the singular orbits-- zeroth order in $\varepsilon$ approximations of the BVP-- are reduced to a nonlinear system algebraic equations that involve all relevant quantities together.   This system of nonlinear algebraic equations accurately draws the physical framework of the problem.
Furthermore, it confirms that all quantities interact with each other, and we will show in this paper quantitatively how some of those transactions occur.

With its extensions to include some of the effects of ion size, this geometric framework has produced several outcomes that are fundamental to ion channel properties \cite{JL12, JLZ15, Liu18, ML19, SL18, ZEL}.
 The interested readers are referred to the papers mentioned above for more details on the GSP framework for PNP and concrete applications to ion channel problems.

%We follow the notations in \cite{ML19} for analytical results where the quantities are all in their dimensionless forms. 
%In addition, 

For simplicity, we use the letters $l$, $r$ and $Q_0$ where
$ l_1=l_2=l$, $r_1=r_2=r$, $Q_2=2Q_0$. 
 Following the framework in \cite{Liu09},
 because of the jumps of the permanent charge $Q(x)$ in \eqref{Q} at $x = a$ and $x = b$, we divide the formation of a singular orbit on the interval $[0, 1]$ into that on three subintervals $[0, a],~[a, b]$ and $[b, 1]$ to convert the boundary value problem to a connecting problem. We denote $C=(c_1,  c_2)^T$ and $J=(J_1, J_2)^T$, and preassign values of $\phi$ and $C$ at $x_a=a$ and $x_b=b$:
\[\phi(x_j)=\phi^{j}\;\mbox{ and }\; C(x_j)=C^{j}\;\mbox{ for }\; j\in \{a,b\}.\]
Now for  $j\in\{l,a,b,r\}$, let $B_j$ be the subsets of the phase space $\bbR^{7}$ defined by
$$
B_j=\Big\{(\phi, u, C, J, w):\;\phi=\phi^{j},\; C=C^{j},\; w=x_j\Big\}.
$$
 
Note that the sets $B_l$ and $B_r$ are associated to the boundary condition in \eqref{1dBV} at $x=0$ and $x=1$ respectively.
 Thus, the  BVP (\ref{1dPNPdim}) and (\ref{1dBV}) is equivalent to the following connecting orbit problem: finding an orbit of (\ref{slow}) from $B_l$ to $B_r$ (See Figure \ref{Fig-PNPdiag}).
The construction would be accomplished by finding first a singular connecting orbit -- a union of limiting slow orbits and limiting fast orbits,
and then applying the exchange lemma to show the existence of a connecting orbit for $\varepsilon>0$ small. For the problem at hand, the construction of a singular orbit consists of one singular connecting orbit from $B_l$ to $B_a$, one from $B_a$ to $B_b$, and one from $B_b$ to $B_r$ with a matching of $(J_1,J_2)$ and $u$ at $x=a$ and $x=b$ (see \cite{Liu09} for details).

%%%%%%%%%%%%%%%%%%%%%%%%%%%%%%%%%%

  \begin{figure}[h]\label{Fig-PNPdiag}
	\centerline{\epsfxsize=5.5in \epsfbox{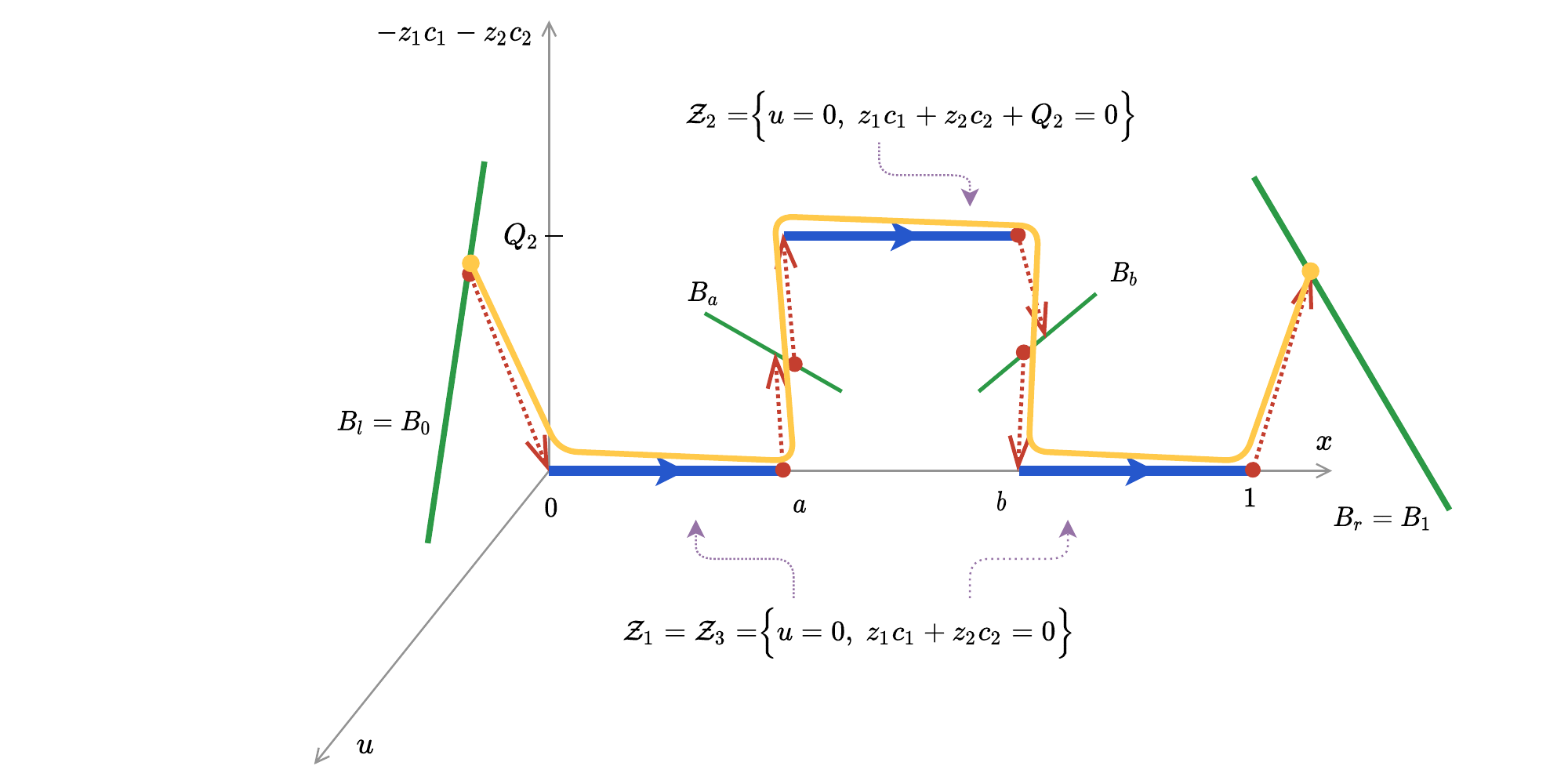} }\vspace*{-.1in}
	\caption{\em An illustration of a singular connecting orbit projected to the space of $(u;  z_1c_1 + z_2c_2; x)$. The solid line from the left boundary $B_l$ to the right boundary $B_r$ is the $O(\varepsilon)$ estimate of the connected problem obtained by Exchange Lemma (see \cite{EL07}). }
	\label{Fig-PNPdiag}
\end{figure}
  
%%%%%%%%%%%%%%%%%%%%%%%%%%%%%%%%%%%%%%%%%%%%%%

By setting $\varepsilon=0$ in system (\ref{slow}), we  get the slow manifold, for $k=1,2,3$,
 \[{\mathcal Z}_k=\Big\{u=0,\; z_1c_1+z_2c_2+ Q_k=0\Big\}.\] 
  In terms of the  independent variable $\xi= x/\varepsilon$, 
  we obtain the fast system of~\eqref{slow}, for $k=1,2$, 
 \begin{align}\label{fast}\begin{split} 
      & \phi'=   u,\;
    u'=  -z_1c_1-z_2c_2 - Q_j- \varepsilon \dfrac{h_w(w)}{h(w)} u ,\\
   & c_k'=-z_kc_ku -\varepsilon \dfrac{J_k}{D_k h(w)}, \quad J'= 0, \quad  w'=\varepsilon,   
 \end{split}
 \end{align}
where   prime denotes the derivative with respect to   $\xi$.
 The limiting fast system is, for $k=1,2$,
\begin{align}\label{limfast}\begin{split} 
     \phi'=&   u,\quad
    u'=  -z_1c_1-z_2c_2  - Q_j, \quad
    c_k'=-z_kc_ku, \quad J'=0,\quad  w'=0.
 \end{split}
 \end{align}

The slow manifold ${\mathcal Z}_k$ is precisely the set of equilibria of (\ref{limfast}) with   $\dim{\mathcal Z}_k=5$.  For the linearization of (\ref{limfast}) at each point on ${\mathcal Z}_kj$,  there are  $5$ zero eigenvalues associated to the tangent space of ${\mathcal Z}_k$ and the other two eigenvalues 
 are $ \pm\sqrt{z_1^2c_1+z_2^2c_2}$.
Thus, ${\mathcal Z}_k$ is normally hyperbolic (\cite{Fen, HPS}). We will denote the stable and unstable manifolds of ${\mathcal Z}_k$ by
$W^s({\mathcal Z}_k)$ and  $W^u({\mathcal Z}_k)$, respectively.  
%%%%%%%%%%%%%%%%%%%%%%%%%%%%%%%
Let $M^{[k-1,+]}$   be the collection of all forward orbits from  $B_{j_1}$ under the flow of \eqref{limfast}
and let  $M^{[k,-]}$   be the collection of all backward orbits from $B_{j_2}$, where $j_1, j_2$ are two of consecutive letters in the set $\{l,a,b,r\}$  (corresponding to $B_l, B_a, B_b, B_r$) where the place of letters is fixed. Then the set of forward orbits from $B_{j_1}$ to the corresponding
${\mathcal Z}_k$ is  $N^{[j_1,+]}=M^{[j_1,+]}\cap W^s({\mathcal Z}_k)$, and the set of backward orbits from $B_{j_2}$ to 
${\mathcal Z}_k$ is $N^{[j_2,-]}=M^{[j_2,-]}\cap W^u({\mathcal Z}_k)$. Therefore, the singular layer $\Gamma^{[j_1,+]}$ at $x_{j_1}$ satisfies  $\Gamma^{[j_1,+]}\subset  N^{[j_1,+]} $ and the singular layer $\Gamma^{[j_2,-]}$ at $x_{j_2}$ satisfies $\Gamma^{[j_2,-]}\subset N^{[j_2,-]}$.
 All the essential geometric objects are explicitly delineated in \cite{Liu09}.

The limiting fast (layer) dynamics preserve electrochemical potentials and do not depend on diffusion constants (\cite{Liu09}). 
Thus, we directly apply the results on the fast dynamics from \cite{ELX15}  and only point out that we consider $n=2$ ion species and need to keep 
$\phi^{a,l}$, $\phi^{a,m}$, $\phi^{b,m}$ and $\phi^{b,r}$ in this paper, while in \cite{ELX15}, the equality of diffusion constants $D_1=D_2$ concludes  $\phi^{a,l}=V$, $\phi^{a,m}=\phi^{b,m}$ (denoted by ${\cal V}^*$ there) and $\phi^{b,r}=0$.
 
  \begin{lem}\label{Lem-fast1}([The fast layer dynamics]) One has, 
  
\noindent  a.) over $x=a$ provides, for $k=1,2$,
  \begin{itemize}
  \item[(i)] relative to $(0,a)$ where $Q_1=0$,
 \[  z_1c_1^{a}e^{z_1(\phi^{a}-\phi^{a,l})}+ z_2c_2^{a}e^{z_2(\phi^{a}-\phi^{a,l})}=0,
 \quad c_k^{a,l}=c_k^{a}e^{z_k( \phi^{a}-\phi^{a,l})};\]
 \item[(ii)] relative to $(a,b)$ where $Q_2\neq 0$,
 \[  z_1c_1^{a}e^{z_1(\phi^{a}-\phi^{a,m})}+ z_2c_2^{a}e^{z_2(\phi^{a}-\phi^{a,m})} + Q_2=0, \quad c_k^{a,m}=c_k^{a}e^{z_k( \phi^{a}-\phi^{a,m})};\]
 \item[(iii)] the matching   $u_-^{a}=u_+^{a}$: \quad
 $c_1^{a,l} + c_2^{a,l} =c_1^{a,m} + c_2^{a,m} +Q_2(\phi^{a}-\phi^{a,m})$;
 \end{itemize}
 
 \noindent b.) over $x=b$ provides, for $k=1,2$,
  \begin{itemize}
  \item[(i)] relative to $(a,b)$ where $Q_2\neq 0$,
 \[  z_1c_1^{b}e^{z_1(\phi^{b}-\phi^{b,m})}+ z_2c_2^{b}e^{z_2(\phi^{b}-\phi^{b,m})} + Q_2=0,
 \quad c_k^{b,m}=c_k^{b}e^{z_k( \phi^{b}-\phi^{b,m})};\]
 \item[(ii)] relative to $(b,1)$ where $Q_3= 0$, 
 \[  z_1c_1^{b}e^{z_1(\phi^{b}-\phi^{b,r})}+ z_2c_2^{b}e^{z_2(\phi^{b}-\phi^{b,r})} =0, \quad c_k^{b,r}=c_k^{b}e^{z_k( \phi^{b}-\phi^{b,r})};\]
 \item[(iii)] the matching   $u_-^{b}=u_+^{b}$: \quad
 $c_1^{b,m} + c_2^{b,m} +Q_2(\phi^{b}-\phi^{b,m}) =c_1^{b,r} + c_2^{b,r}.$
 \end{itemize}
 \end{lem} 
%%%%%%%%%%%%%%%%%%%%%%%%%%%%
As stated in the Introduction, the degeneracy of equal diffusion coefficients arises in the slow dynamics. Diffusion is a phenomenon in which the spatial distribution of solute particles varies due to their potential energy. It is a standard process that works to remove differences in concentration and eventually turns a given mixture to a uniform structure state. The diffusion of uncharged particles can be described by the Fick's first law \cite{Fick}  through the equation $\partial_t c= \mathcal{D}\partial^2_{xx} c$, where $c$ is the concentration, $\mathcal{D}$ is the diffusion constant, and $t$ is time.
 The measurement of diffusion constants frequently involves measuring collections of simultaneous values of $t$, $c$, and  $x$. These evaluated values are then implemented to a solution of Fick's law to achieve the diffusion constants. 
There are many approaches and techniques to maintain diffusion constants of ions in aqueous solutions \cite{BF80, BHW66, GA82, LE14}.

Some kinds of selectivity depend on the non-equality of diffusion coefficients. Besides, many electrical phenomena wholly vanish. That is, the liquid coupling is zero when diffusion constants are identical. 
 Therefore, the equal diffusion constants case is degenerate.
On the other hand, everything becomes much more complicated (at least mathematically) when the diffusion constants are unequal.

Some ionic species' diffusion constants may alter from one technique to another, even when all other parameters are held untouched. However, experimental measurements are directed under isothermal circumstances to keep away from a deviation of $\mathcal{D}$ values.

%%%%%%%%%%%%%%%%%%%%%%%%%%%%%%%%%%%%%%%%%

In this part, we briefly recall, with different notation, the analysis on the slow layer from Section 2  of the paper \cite{ML19}. For zero current
  $I=z_1J_1 + z_2J_2 =0$ (so $J_1=J_2$), and further assumption $z_1=-z_2=1$, one has
$
 \frac{J_1}{D_1} - \frac{J_2}{D_2} = \frac{D_2 - D_1}{D_1D_2}J_1.
$
%%%%%%%%%%%%%%%%%%%%%%%%%%
 Applying zero current condition in above, the limiting slow system becomes(\cite{ML19}),
\begin{align}\label{Dslow1}
\begin{split}
  \dot \phi=&  -\frac{(D_2-D_1)J_1}{D_1 D_2 h(\omega) (2c_1 +Q_j)},\quad
    \dot c_1  =-\frac{(D_2 + D_1) c_1+D_2Q_j }{D_1D_2 h(\omega) (2c_1 +Q_j)} J_1,\quad
\dot J_1=0,\quad \dot w=1.
\end{split}
 \end{align}
%%%%%%%%%%%%%%%%%%%%%%%%%%%%%
 The slow system \eqref{Dslow1} on $(0,a)$ and $(b,1)$  where $Q_1=Q_3=0$ and on $(a,b)$  where $Q_2\neq 0$ will result in the following Lemma.
 
 \begin{lem}\label{Lem-slow1}
 The slow dynamics over each interval (slow manifold) is as follows:
 
 \noindent (a) Over $(0,a)$ with $c_1(x) -c_2(x)=0$ the slow dynamics system gives,
\begin{align*}
c_1^{a,l} = l_1 - \frac{D_1+D_2 }{ 2D_1 D_2} J_1 H(a),\quad
\phi^{a,l} = V - \frac{D_1 - D_2}{D_1+D_2} \ln \dfrac{c_1^{a,l}}{l_1}.
\end{align*}
 
 \noindent (b) Over $(a,b)$ with $c_1(x) -c_2(x)+Q_2=0$ the slow dynamics system gives,
\begin{align*} 
&\phi^{b,m}= \phi^{a,m} + \dfrac{D_1 - D_2}{D_1 D_2} J_1y^*,\\
& c_{1}^{b,m}= e^{-\frac{D_1+D_2 }{D_1 D_2} J_1y^*}c_1^{a,m}+\dfrac{D_2Q_2}{D_1 + D_2}\big( e^{-\frac{D_1+ D_2}{D_1 D_2} J_1y^*}-1 \big),\\
&J_1= -D_1D_2\dfrac{2\big(c_{1}^{b,m} - c_{1}^{a,m} \big)-(\phi^{b,m} - \phi^{a,m})Q_2}{(D_1 + D_2)\big(H(b) - H(a) \big) } .
\end{align*}
 
 \noindent (c) Over $(b,1)$ with $c_1(x) -c_2(x)=0$ the slow dynamics system gives,
\begin{equation*}
 c_1^{b,r} = r_1 + \frac{D_1 +D_2  }{2D_1 D_2} J_1 \big(H(1) - H(b) \big),\quad
\phi^{b,r} =  \frac{D_1 -D_2}{D_1 + D_2} \ln \dfrac{r_1}{c_1^{b,r}}.
 \end{equation*}
 \end{lem}
 
 \subsection*{\underline{Matching for Zero-current and Singular Orbits on $[0,1]$.}}
The last step to build a connecting orbit over the whole interval 
$[0,1]$ is to match the three singular orbits from the Lemmas
 \ref{Lem-fast1} and \ref{Lem-slow1} at the points $x=a$ and $x=b$. 
The matching conditions are $u_-^{a}=u_+^{a},~u_-^{b}=u_+^{b}$, and that $J_1(=J_2)$ needs to be equal on all three subintervals. Thus,
%Recall that we analyze the case where $n=2$ (two ion species) with $z_1 = -z_2=1$. Thus,
%%%%%%%%%%%%%%%%%%%%%%%%%%%%%%%%%%%%%%
\begin{equation}\label{Matching}
\begin{aligned}
&{ c_1^{a}e^{(\phi^{a}-\phi^{a,l})}-c_2^{a}e^{-(\phi^{a}-\phi^{a,l})}= c_1^{b}e^{(\phi^{b}-\phi^{b,r})}- c_2^{b}e^{-(\phi^{b}-\phi^{b,r})} = 0,}\\
& { c_1^{a}e^{( \phi^{a}-\phi^{a,m})}- c_2^{a}e^{-( \phi^{a}-\phi^{a,m})}+ Q_2=0,}\\
&  { c_1^{b}e^{( \phi^{b}-\phi^{b,m})}- c_2^{b}e^{-( \phi^{b}-\phi^{b,m})}+ Q_2=0,}\\
&{2c_1^{a,l}}=c_1^{a}e^{( \phi^{a}-\phi^{a,m})} + c_2^{a}e^{{-}( \phi^{a}-\phi^{a,m})} +Q_2(\phi^{a}-\phi^{a,m}),\\
&{2c_1^{b,r}}=c_1^{b}e^{( \phi^{b}-\phi^{b,m})} + c_2^{b}e^{-( \phi^{b}-\phi^{b,m})} +Q_2(\phi^{b}-\phi^{b,m}),\\
&J_1=J_2 = -\dfrac{2D_1D_2(c_1^{a,l}- l_1)}{(D_1+D_2 ) H(a)} = -\dfrac{2D_1D_2(r_1- c_1^{b,r})}{(D_1 + D_2  )(H(1) - H(b))},\\
&\hspace*{.5in} = -D_1D_2\dfrac{2(c_1^{b,m}- c_1^{a,m}) - (\phi^{b,m} - \phi^{a,m})Q_2}{(D_1 +D_2 )(H(b)-H(a))},\\
&\phi^{b,m}= \phi^{a,m} + \dfrac{D_1 - D_2}{D_1 D_2} J_1y^*,\\
&c_{1}^{b,m}= e^{-\frac{D_1 + D_2}{D_1 D_2} J_1y^*}c_1^{a,m}+\dfrac{D_2 Q_2}{(D_1 +D_2) }\Big( e^{-\frac{D_1+ D_2}{D_1 D_2} J_1y^*}-1 \Big).
\end{aligned}
\end{equation}
%where,
%\begin{equation}\label{Matching2}
%\begin{aligned}
%c_1^{a,l}=&c_1^{a}e^{z_1( \phi^{a}-\phi^{a,l})} = \sqrt{c_1^{a} c_2^{a}} ,\quad
%c_2^{a,l}= c_2^{a}e^{z_2( \phi^{a}-\phi^{a,l})} = \sqrt{c_1^{a} c_2^{a}},\\
%c_1^{b,r}=& c_1^{b}e^{z_1( \phi^{b}-\phi^{b,r})} =  \sqrt{c_1^{b} c_2^{b}},\quad
%c_2^{b,r}=  c_2^{b}e^{z_2( \phi^{b}-\phi^{b,r})} ={ \sqrt{c_1^{b} c_2^{b}}},\\
%c_1^{a,m}=&c_1^{a}e^{z_1( \phi^{a}-\phi^{a,m})}, \quad c_2^{a,m}=c_2^{a}e^{z_2( \phi^{a}-\phi^{a,m})},\\
%c_1^{b,m}=&c_1^{b}e^{z_1( \phi^{b}-\phi^{b,m})}, \quad c_2^{b,m}=c_2^{b}e^{z_2( \phi^{b}-\phi^{b,m})},\\
%\phi^{a,l} =& V - \frac{D_1 - D_2}{z_1(D_1  + D_2)} \ln \frac{c_1^{a,l}}{l_1},\quad \phi^{b,r} 
%=  \frac{D_1 - D_2}{z_1(D_1 +D_2)} \ln \frac{r_1}{c_1^{b,r}}.
%\end{aligned}
%\end{equation}

\begin{rem} In \eqref{Matching}, the unknowns are: $\phi^{a},~\phi^{b},~ c_1^{a},~c_2^{a},~ c_1^{b},~c_2^{b},~J_1,~\phi^{a,m},~\phi^{b,m},~y^*$ and $Q_2$  that is, there are eleven unknowns that matches the total number of equations on \eqref{Matching}. \qed
\end{rem}

%%%%%%%%%%%%%%%%%%%%%%%%%%%%%%%%%%%%%%%%%

%%%%%%%%%%%%%%%%%%%%%%%%%%%%%%%%%%%%%%%%%%

%%%%%%%%%%%%%%%%%%%%%%%%%%%%%%%%%%%%%%%%%%%%

%%%%%%%%%%%%%%%%%%%%%%%%%%%%%%%%%%%%%%%%%
%%%%%%%%%%%%%%%%%%%%%%%%%%%%%%%%%%%%%%%%%
%\newpage

\section{Zero-Current Problems: Geometric Mean of Concentrations and Reversal Permanent Charge.}\label{sec-ZeroCur}
\setcounter{equation}{0}
This section investigates how some quantities, such as boundary concentrations and electric potential, work to make the current reverse. We analyze the outcomes for various diffusion constants to reveal the impacts of diffusion constants on the geometric mean of concentrations and reversal permanent charge and expose the opposites.

%%%%%%%%%%%%%%%%%%%%%%%%%%%%%%%%%%%%%%%%%%%%%%%

We encountered a nonlinear governing system in \eqref{Matching} that is difficult to analyze as it is.
In \cite{ELX15}, for equal diffusion constants, $D_k$'s, the investigation on  reversal permanent charges implemented for a general $n$. However, with the general $D_k$'s, the difficulty intensifies, as shown in \cite{ML19}, even for the case we will address in this work, where there are two ion species, i.e., $n=2$ with $z_1=-z_2=1$.

%%%%%%%%%%%%%%%%%%%%%%%%%%%%%%%%%%%%%%%%%%%%%

In \cite{EL07} and \cite{ML19}, without and with zero current assumptions, respectively, the authors offered two intermediate variables that support a notable reduction of the matching (\ref{Matching}):
\begin{equation}\label{A}
A= \sqrt{c_1(a)c_2(a)}, \quad B= \sqrt{c_1(b)c_2(b)}. 
\end{equation}
The variables $A$ and $B$ are the geometric mean of concentrations at $x=a$ and $x=b$ respectively.
One may consider $B$ as a function of $A$ rather than an independent variable. In fact, it can be seen that  $B=B(A) = \dfrac{1-\beta}{\alpha}(l-A)+r$. 
We can now reduce the Matching system \eqref{Matching} to a nonlinear system with two equations and two unknowns. We skip the redued step here, for zero current ${I}= J_1 -J_2 =0$, as  illustrated in detail in the Appendix of \cite{ML19}.

\begin{align}\label{G1G2Sys}
G_1(Q_0, A, \delta )=V\;\mbox{ and }\; G_2(Q_0,A,\delta )=0,
\end{align}

\noindent
where

\begin{equation}\label{G}
\begin{aligned}
G_1(Q_0, A, \delta )=& \delta  \Big(\ln\dfrac{S_a + \delta Q_0}{S_b + \delta  Q_0} + \ln\dfrac{l}{r}\Big)  - (1+\delta )\ln \dfrac{A}{B} + \ln \dfrac{S_a -Q_0}{S_b -Q_0},\\
G_2(Q_0, A, \delta )=& \delta Q_0\ln\dfrac{S_a+\delta Q_0}{S_b+\delta Q_0}-N.
\end{aligned}
\end{equation}	 
%%%%%%%%%%%%%%%%%%%%%%%%%%%%%%%%%%%%%%%%%
Furthermore,
\begin{equation}\label{Nab}
   S_a= \sqrt{Q_0^2+A^2}, \quad S_b= \sqrt{Q_0^2+B^2}, 
 \quad N = A-l  + S_a -S_b,
\end{equation}
and

\begin{equation}\label{al-be-thet}
\delta =\frac{D_2- D_1}{D_2 + D_1},\quad H(x)=\int_0^x\frac{1}{h(s)}ds, \quad \alpha=\frac{H(a)}{H(1)}, \quad \beta=\frac{H(b)}{H(1)}.
\end{equation}

 The function $H(x)$ is the ratio of the length with the cross-section area of the portion of the channel over $[0,x]$. The quantity $H(x)$ initially has roots in Ohm law for a uniform resistor's resistance. The quantities $\alpha$ and $\beta$, together with $Q_0$, are critical characteristics for the shape and permanent charge of the channel structure \cite{JLZ15}.

%%%%%%%%%%%%%%%%%%%%%%%%%%
Partial derivatives of $G_1$ and $G_2$ with respect to $Q_0$ and $A$ are,
\begin{align}\label{parG1G2}\begin{split}
\partial_{A}G_1(A, Q_0,\delta) =& (1-\delta^2)Q_0\Big(\dfrac{1}{A(S_a+\delta Q_0)} +\frac{1-\beta}{\alpha} \dfrac{1}{B(S_b+\delta Q_0)}\Big),\\
\partial_{Q_0}G_1(A, Q_0,\delta) =&  \dfrac{(1-\delta^2)(S_a-S_b)}{(S_a+\delta Q_0)(S_b+\delta Q_0)} ,\\
\partial_{\delta} G_1 (A, Q_0,\delta) =&\ln \dfrac{S_a+\delta Q_0}{S_b+\delta Q_0}   + \dfrac{\delta Q_0}{S_a+\delta Q_0}-\dfrac{\delta Q_0}{S_b+\delta Q_0} +\ln \dfrac{l}{r} - \ln \dfrac{A}{B},\\
 \partial_A G_2 (A, Q_0,\delta)=& -\frac{1-\beta}{\alpha}\dfrac{B}{S_b+\delta Q_0} - \dfrac{A}{S_a+\delta Q_0} - \dfrac{\beta - \alpha}{\alpha},\\
\partial_{Q_0} G_2(A, Q_0,\delta) =&\delta \ln\dfrac{S_a+\delta Q_0}{S_b+\delta Q_0}+\dfrac{(1-\delta^2)Q_0(S_a-S_b)}{(S_a+\delta Q_0)(S_b+\delta Q_0)},\\
\partial_{\delta} G_2 (A, Q_0,\delta) =&Q_0\ln \dfrac{S_a+\delta Q_0}{S_b+\delta Q_0}   + \dfrac{\delta Q_0^2}{S_a+\delta Q_0}-\dfrac{\delta Q_0^2}{S_b+\delta Q_0}.
\end{split}
\end{align} 
 %%%%%%%%%%%%%%%%%%%%%%%%%%%%%%%%%%
 
It follows directly from (\ref{parG1G2}) that $G_1(Q_0, A, \delta)$ is increasing (decreasing) in $A$ as $Q_0$ is increasing (decreasing), and it is increasing (decreasing) in $Q_0$ or $\delta$ when $l-r$ increases (decreases). Moreover, $G_2$ always decreases in $A$ regardless of other values; however, its behavior with respect to $Q_0$ and $\delta$ is complicated. The preceding observations have been confirmed in Lemma 3.2 of \cite{ML19} in detail; however, the authors of that paper did not address how $G_2$ acts (with respect to $Q_0$ or $\delta$) when $\delta Q_0< 0$.
We develop and complete the Lemma in the following Theorem to serve us in the next tracks. In particular, it will help us later in Section \ref{revPC}, to establish the uniqueness of reversal potential.
%%%%%%%%%%%%%%%%%%%%%%%%%%%%%%%%%%%%%%%%%%
\begin{thm}\label{Thm-G2} Assume that diffusion constants $D_1$ and $D_2$ are fixed: \medskip

\noindent (a) if $D_1 <D_2$, i.e. $\delta >0$, then there exists a  $Q_0^-<0$, so that $\partial_{Q_0} G_2(Q_0, A, \delta)$ has the same sign as that of $l-r$ when $Q_0>Q_0^-$, but it has the opposite sign when $Q_0 < Q_0^-$, and \medskip

\noindent(b) if $D_1 > D_2$, i.e. $\delta <0$, then there exists a  $Q_0^+>0$, so that $\partial_{Q_0} G_2(Q_0, A, \delta)$ has the opposite sign as that of $l-r$ when $Q_0<Q_0^+$, but it has the same sign when $Q_0 > Q_0^+$.
\end{thm}
\begin{proof} We prove (a). The other statement is similar. 
It follows from \eqref{G} that $\partial_{Q_0}G_2$ for small (large) values of $Q_0<0$ has the same (opposite) sign as that of $l-r$. And for any values of $Q_0>0$, it always has the same sign as that of $l-r$.
%%%%%%%%%%%%%%%%%%%%%%%%%%%%%%%
On the other hand, taking one more derivative, $\partial^2_{Q_0}G_2$ can be written (after making common factors) in form of $\dfrac{P_1(Q_0)}{P_2(Q_0)}$ where $P_1$ and $P_2$ are two polynomials (in $Q_0$) of degrees $2$ and $4$ respectively. It is easy to see that $P_2(Q_0) \neq 0$ for any $Q_0$, and $P_1(Q_0)$ has at most two roots. Thus, $\partial^2_{Q_0}G_2$ may change sign at most twice. This completes the proof.
\end{proof}
%%%%%%%%%%%%%%%%%%%%%%%%%%%%%%%%%%%%%%
\begin{cor}\label{Cor-Q}
A direct conclusion of Theorem \ref{Thm-G2} is that $Q_0^-$ and $Q_0^+$ take, respectively, the absolute minimum and maximum of $\partial_{Q_0}G_2$ for each case.
\end{cor}
%%%%%%%%%%%%%%%%%%%%%%%%%%%%%%%%%%%%%%%
\begin{rem}
 We will see later in Theorem \ref{Thm-A} that $\partial_{Q_0}G_2$ and $\partial_{Q_0}A$ have the same behavior. Hence, Figure \ref{Fig-AvsQ} for $A$ may also be trusted to help the reader understand Theorem \ref{Thm-G2}.
 \qed
\end{rem}

%%%%%%%%%%%%%%%%%%%%%%%%%%%%%%%%%%%%
\subsection{Results on Geometric Mean of Concentrations.}
We established the geometric mean of concentration on \eqref{A} and stated why we need to specify this parameter. The authors of \cite{ML19} prove that, for fixed $Q_0, D_1$ and $ D_2$, one can solve for $A$ from $G_2(Q_0, A,\delta )=0$, where $G_2$ is defined in (\ref{G}). Moreover, $A(0,\delta ) = (1-\alpha)l + \alpha r$  and $\lim_{Q_0 \to \pm \infty} A(Q_0,\delta ) = l$. Then, they showed that 
$l, r, A$ and $B$ satisfy one the following conditions: $l<A(Q_0,\delta )<B(Q_0,\delta )<r$ or 
 $ l>A(Q_0,\delta )>B(Q_0,\delta )>r$. If   $\delta Q_0\ge 0$, then   $\partial_{Q_0}A(Q_0,\delta )$ has the same sign as that of $(l-r)Q_0$.

 We now present a few more features of the geometric mean of concentrations, $ A $, through the following Theorems and a Corollary. The same results can be achieved for the other parameter, $B$. 
The following Lemma is directly obtained from the equations \eqref{G1G2Sys}--\eqref{al-be-thet}.
 
\begin{lem}\label{lem-sgnN} For any given $Q_0 \in \mathbb{R}$ and for any $A\in \big(0, A_M\big) $ where $A_M =l + \frac{\alpha }{1-\beta}r$ is the maximum value of $A$, the functions $G_1(Q_0,A, \delta)$ and $G_2(Q_0,A, \delta)$ in \eqref{G} are continuous. Furthermore, 
$\partial_A N >0$, $\partial_A B <0$.
\end{lem}

\begin{rem}
It might be deduced, incorrectly from Corollary \ref{Cor-Q} and Lemma \ref{lem-sgnN}, that $A_M$ is equal to $A(Q_0)$ evaluated at $Q_0=Q^-$, or $Q_0=Q^+$. However, the function $A(Q_0)$ is not surjective necessarily, and in general, $A(Q_0^{\pm})\leq A_M$ . In next Theorem, we discuss about the bounds of the geometric mean of concentration in more details.
\end{rem}
%%%%%%%%%%%%%%%%%%%%%%%%%%%%%%%%%%%%%%
\begin{thm}\label{Thm-A}
Assume that diffusion constants $D_1$ and $D_2$ are fixed. Let $Q_0^-<0$, $Q_0^+>0$ be  as defined in Theorem \ref{Thm-G2}. Subsequently, \medskip

\noindent (a) if $D_1 <D_2$, i.e. $\delta >0$, then  $\partial_{Q_0} A(Q_0, \delta)$ has the same sign as that of $l-r$ when $Q_0>Q_0^-$, but it has the opposite sign when $Q_0 < Q_0^-$, and \medskip

\noindent(b) if $D_1 > D_2$, i.e. $\delta <0$, then  $\partial_{Q_0} A(Q_0,\delta)$ has the opposite sign as that of $l-r$ when $Q_0<Q_0^+$, but it has the same sign when $Q_0 > Q_0^+$.
\end{thm}
%%%%%%%%%%%%%%%%%%%%%%%%%%%%%
\begin{proof} This is the direct conclusion of Theorem \ref{Thm-G2} with the facts that $\partial_{Q_0}A=-\partial_{Q_0}G_2/\partial_AG_2$ and $\partial_AG_2<0$ for any $Q_0$.
\end{proof}
%%%%%%%%%%%%%%%%%%%%%%%%%%%%%
\begin{rem}
{\em Note that $Q_0^-$ and $Q_0^+$  in Theorems \ref{Thm-G2} and \ref{Thm-A} are the same.   Similarly, we emphasize that the same conclusions can be established for the case with $D_1 > D_2$.}
\end{rem}

%%%%%%%%%%%%%%%%%%%%%%%%%%%%%%%%%%%%%
\begin{thm}\label{thmA(Q)} For any given $Q=Q_0$ there is a unique $A(Q_0)$ such that $G_2(A(Q_0), Q_0)=0$. Furthermore, $\lim_{Q_0 \to \pm \infty} A(Q_0) = l$, and $A=A(Q_0)$ satisfies the following,
  
\noindent (a)~if $l<r$ then $l \leq A(Q_0) \leq A(Q_0^{\pm})$, \qquad (b)~if $l>r$ then $A(Q_0^{\pm})\leq A(Q_0) \leq l$,\\
(c)~if $l=r$ then  $A(Q_0) =l =r = A(Q_0^{\pm})$,\\
where $A(Q_0^{\pm})= A(Q^{+})$ or $A(Q_0^{\pm})= A(Q^{-})$ depending on the sign of diffusion constants; also, $Q_0^{+}$, and $Q_0^{-}$ were defined in Theorem \ref{Thm-A}.
\end{thm}
\begin{proof}
The uniqueness of $A=A(Q_0)$ has been proven in \cite{ML19}.
To show the limit, since 
$\displaystyle\lim_{Q_0\to \pm \infty} S_a(Q_0)=S_b(Q_0),$
 it follows from $G_2(A(Q_0),Q_0)=0$ that,
$$
\begin{aligned}
\delta\lim_{Q_0\to \pm \infty}Q_0\ln\dfrac{S_a + p Q_0}{S_b + pQ_0}=& \dfrac{\beta - \alpha}{\alpha} \lim_{Q_0\to \pm \infty} (A-l).
\end{aligned}
$$
On the other hand,
$$
\begin{aligned}
\lim_{Q_0\to \pm \infty}Q_0 \ln\dfrac{S_a+pQ_0}{S_b+ pQ_0} =& -\lim_{Q_0\to \pm \infty}\dfrac{(\frac{Q_0}{S_a}+p)(S_b+p Q_0)-(\frac{Q_0}{S_b}+p)(S_a+p Q_0)}{(S_a+p Q_0)(S_b+p Q_0)} Q_0^2=0.
\end{aligned}
$$
Thus, $\lim_{Q_0 \to \pm \infty} A(Q_0) = l$. The parts (a)--(c) concludes from and Theorem \ref{Thm-A}.
\end{proof}

%%%%%%%%%%%%%%%%%%%%%%%%%%%%%%%%%%%%%%
\begin{cor}\label{corA}  
The geometric mean of concentrations $A=A(Q_0)$ is always finite and bounded between $l$ and $A(Q_0^{\pm})$ for any values of $Q_0$. Similarly, the other geometric mean of concentrations $B(Q_0)$ is unique and  bounded between $r$ and $B(Q_0^{\pm})$. Besides,
$ \lim_{Q_0 \to \pm \infty} B(Q_0) = r.$ Furthermore, if $A$ is surjective, then $A(Q_0^{\pm}=A_M$ where $A_M$ was defined in Lemma \ref{lem-sgnN}; and $B(Q_0^{\pm}=B_M$ correspondingly.
\end{cor}

%%%%%%%%%%%%%%%%%%%%%%%%%%%%%%%%%%%%%%%%%

Recall that the quantities $B=B(A)$ and $N=N(A)$ defined in \eqref{A} and \eqref{Nab}. One can simply see that $B-A$ is decreasing over domain of $A$. Besides, $B-A>0$ when $A\to 0^+$ (consequently $B>0$), and $B-A<0$ when $A\to A_M^-$ (consequently $B\to 0^+$).
Therefore, there exists some  $A^*\in (0, A_M)$  for which $B^*:=B(A^*)=  A^*$. Moreover, Lemma \ref{lem-sgnN} along with $N(0^+)<0, N(A_M^-)>0$ conclude that there exist an $\hat{A}(Q_0)\in (0, A_M)$  such that $N(\hat{A})=0$. The corresponding quantities $\hat{B}$ and $B^*$ are similarly defined.

%%%%%%%%%%%%%%%%%%%%%%%

\begin{lem}\label{lem-A*} For any given $Q_0$ one has,

\noindent(i)~If $l<r$ then $l < \hat{A} < A^*=B^* < \hat{B} < r$,\qquad 
(ii)~If $l>r$ then $l > \hat{A} > A^*=B^* > \hat{B} > r$,\\
(iii)~If $l=r$ then $\hat{A}= \hat{B} = A^* =B^* =l =r $.
\end{lem}
\begin{proof}
Suppose $r <l$. Since $N(A=l) > 0,~N(A=\hat{A})=0$ and $N$ is increasing with respect to $A$ then we get $\hat{A}<l$. Now, it follows from definition of $B$, $N(\hat{A})=0$ and $\hat{A}<l$ that $r<\hat{B}<\hat{A}$.\\
Now, set $f(A):= B - A.$ Then,  $f(\hat{A})= \hat{B}- \hat{A}<0$ and 
$$
f(\hat{B})= \Big(\dfrac{1-\beta}{\alpha}(l-\hat{B}) + r \Big) - \Big(\dfrac{1-\beta}{\alpha}(l-\hat{A}) + r \Big) = \dfrac{1-\beta}{\alpha} (\hat{A} - \hat{B}) >0.
$$
Thus, it follows from $\partial_Af<0$  and  $f(A^*)=0$ that $\hat{B} < A^*=B^* <\hat{A}$.
The other cases are similar.
\end{proof}

\begin{thm}\label{Thm-Asym}
For equal diffusion constant $D_1=D_2$, the solution $A(Q_0,0)$ of $G_2(A(Q_0),Q_0,0)=0$ is symmetric with respect to $Q_0$; but for any $(Q_0,\delta)$ where $\delta\neq 0$, i.e. $D_1 \neq D_2$, the function $A(Q_0,\delta)$ is non-symmetric with respect to $Q_0$.
\end{thm}
\begin{proof}
It directly follows from $G_2=0$ in  \eqref{G1G2Sys} and \eqref{G}.
\end{proof}

Figures \ref{Fig-AvsQ} and \ref{Fig-AvsV} confirms the discussions in Theorems \ref{Thm-A}, \ref{thmA(Q)} and \ref{Thm-Asym}.
 In what follows, numerical simulations are conducted with the help of analysis on system (\ref{G1G2Sys}). The combination of numerics and analysis gives a better understanding of the zero-current problems and compliments some analytical results obtained in \cite{ML19}.
 For our numerical simulations, we assume that $h(x)= k$ for any $x\in [a,b]$ with $a=1/3, b=2/3$ in the right panel in Figure \ref{Fig-hx}, where $ 0 < k < 1$. We further suppose that $h(0)= h(1)=1$, and assume $h(x)$ is approximated by two linear functions over the non-constant intervals $(0,1/3)$ and $(2/3,1)$. Thus, $h(x)$ is defined as a piece-wise linear function over $[0,1]$.
%\begin{equation}\label{h-linearized}
%h(x)= \left\{
%\begin{aligned}
%& 1+\dfrac{k-1}{a}x, \qquad \quad ~~ 0 \leq x \leq a,\\
%&\qquad k, \qquad \qquad \qquad ~~ a \leq x \leq b,\\
%& 1+\dfrac{1-k}{1-b}(x-1), \quad b \leq x \leq 1,\\
%\end{aligned}\right.
%\end{equation}
%%%%%%%%%%%%%%%%%%%%%%%%%%%%%%%%%%%%%
Then it follows from definition of $H(x)$ in \eqref{al-be-thet} that
%\begin{equation}
%H(x)= \left\{
%\begin{aligned}
%&\frac{a}{k-1}\ln \big(1+\dfrac{k-1}{a}x\big) , \hspace*{2.32in} 0 \leq x \leq a,\\
%& \frac{a \ln k}{k-1} + \frac{x}{k}-\frac{a}{k} , \hspace*{2.76in} a \leq x \leq b,\\
%& \frac{(1-b+a)\ln k}{k-1} +\frac{b-a}{k} +\frac{1-b}{1-k}\ln \big(1+\frac{1-k}{1-b}(x-1)\big),~~ \quad b\leq x \leq 1.\\
%\end{aligned}\right.
%\end{equation}
%%%%%%%%%%%%%%%%%%%%%%%%
%Therefore,
$$
H(0)=0, \quad H(a)=\dfrac{a\ln k}{k-1}, \quad H(b)=\dfrac{a\ln k}{k-1} + \dfrac{b-a}{k}, \quad H(1)=\dfrac{(1-b+a)\ln k}{k-1} + \dfrac{b-a}{k}.
$$
%%%%%%%%%%%%%%%%%%%%%%%%%%%%%%%%%%%%%
Thus, for $0 < k <1$,
\begin{equation}\label{abk}
\begin{aligned}
\alpha =& \alpha (a,b,k) = \dfrac{ak\ln k}{(1-b+a)k\ln k + (k-1)(b-a)},  \\
\beta = & \beta (a,b,k) = \dfrac{ak\ln k+(k-1)(b-a)}{(1-b+a)k\ln k + (k-1)(b-a)}.
\end{aligned}
\end{equation}

\noindent Since $0<a<b<1$, it follows from above that $ 0<\alpha < \beta$. The following Theorem is the direct conclusion of \eqref{abk}.
%%%%%%%%%%%%%%%%%%%%%%%%%%%%%
\begin{thm}\label{thm-abk}
	One has,\vspace*{-.05in}\\
(i) if $a \to 0^+$, then $\alpha \to 0^+, \beta \to \dfrac{(k-1)b}{(1-b)k\ln k + (k-1)b} $. For this case,\vspace*{-.05in}
		\begin{center}
			 $k\to 0^+$ if and only if $\beta \to 1^-$, and  $k\to 1^-$ if and only if $\beta \to b$;
		\end{center}\vspace*{-.05in}
	\noindent (ii) if $b \to 1^-$, then $\beta \to 1^-, \alpha \to \dfrac{ak\ln k}{k-1} $. For this case,
		\begin{center}\vspace*{-.05in}
			 $k\to 0^+$ if and only if $\alpha \to 0^+$, and
			$k\to 1^-$ if and only if $\alpha \to a$.
			\end{center}\vspace*{-.05in}
Furthermore, when $a$ and $b$ are close, one obtains that all values  $\alpha, \beta, a$ and $b$ are close too.
\end{thm}
%%%%%%%%%%%%%%%%%%%%%%%%%%%%%%%%%%%%%%%%%%%%%%%%
\begin{figure}[h]
\centerline{\epsfxsize=3.0in
\epsfbox{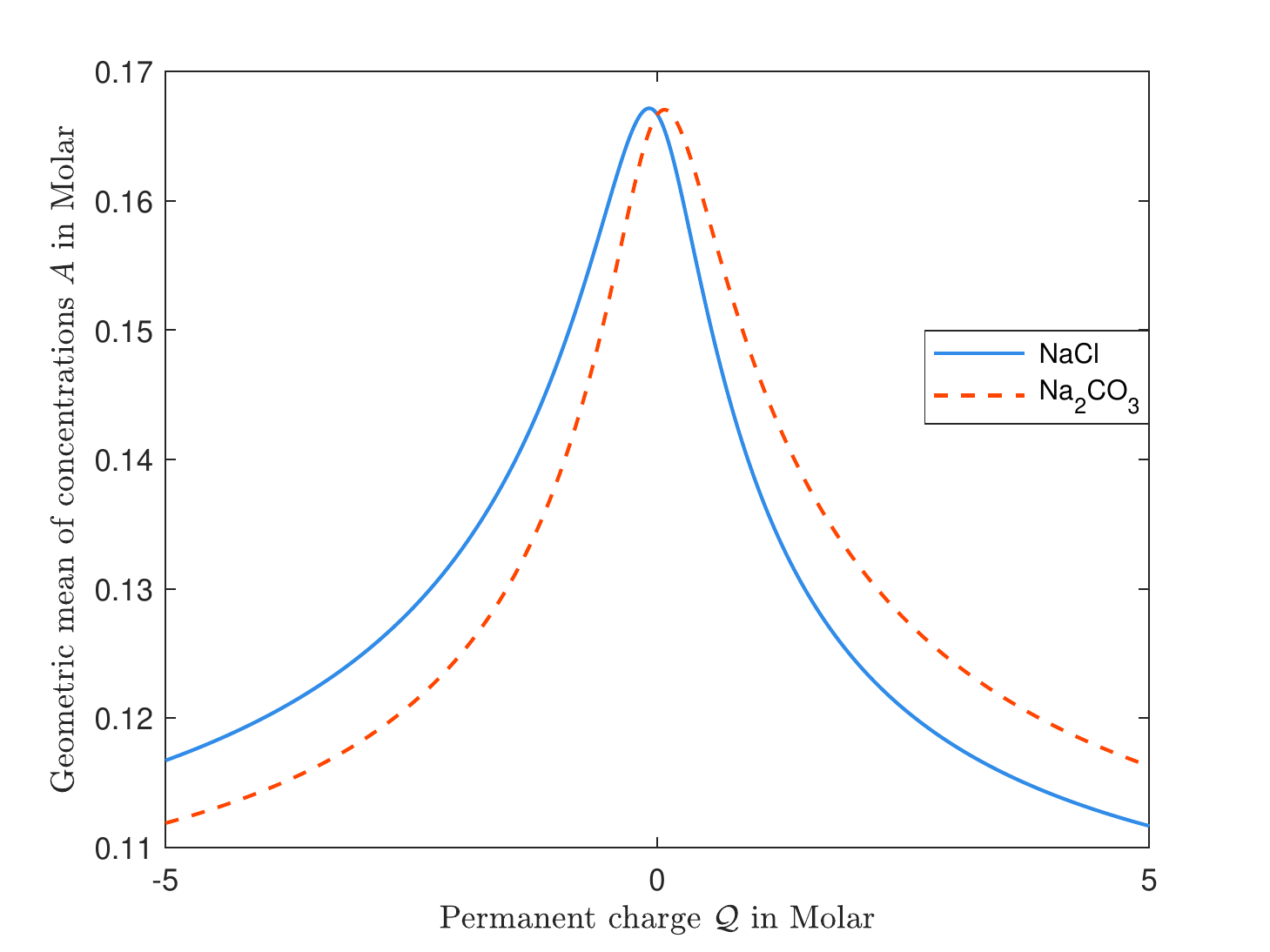} 
\epsfxsize=3.0in \epsfbox{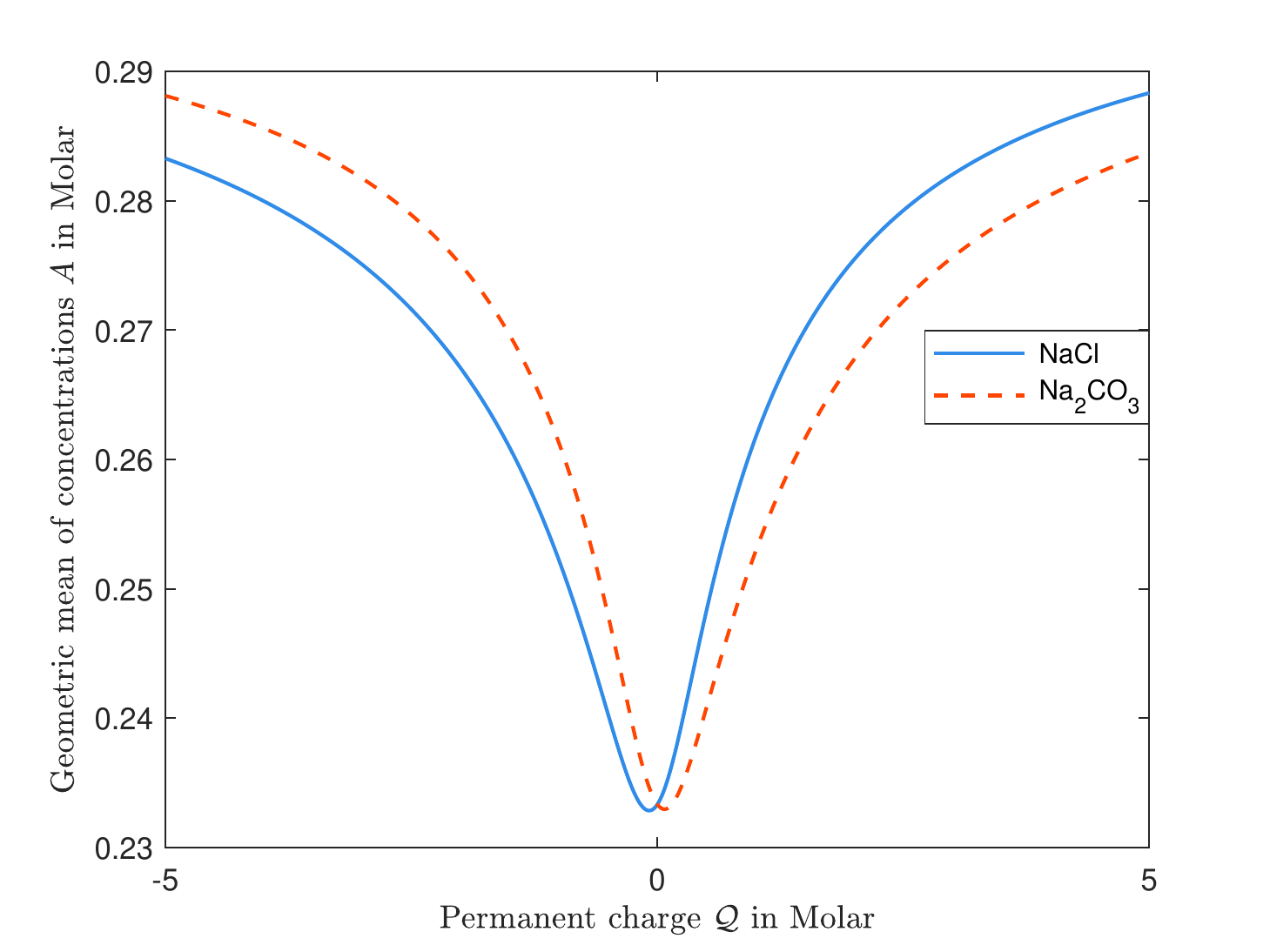} 
}
\caption{\em  The function $A(Q_0,\delta)$  for two pairs of  $(D_1,D_2)$ corresponding to $NaCl$ and $Na_2CO_3$, for various values of $\mathcal{Q}$: left panel for $L=0.1 M, R=0.3 M$; right panel for $L=0.3 M, R=0.1 M$.  }  
\label{Fig-AvsQ}
\end{figure}

\begin{figure}[h]
\centerline{\epsfxsize=3.0in
\epsfbox{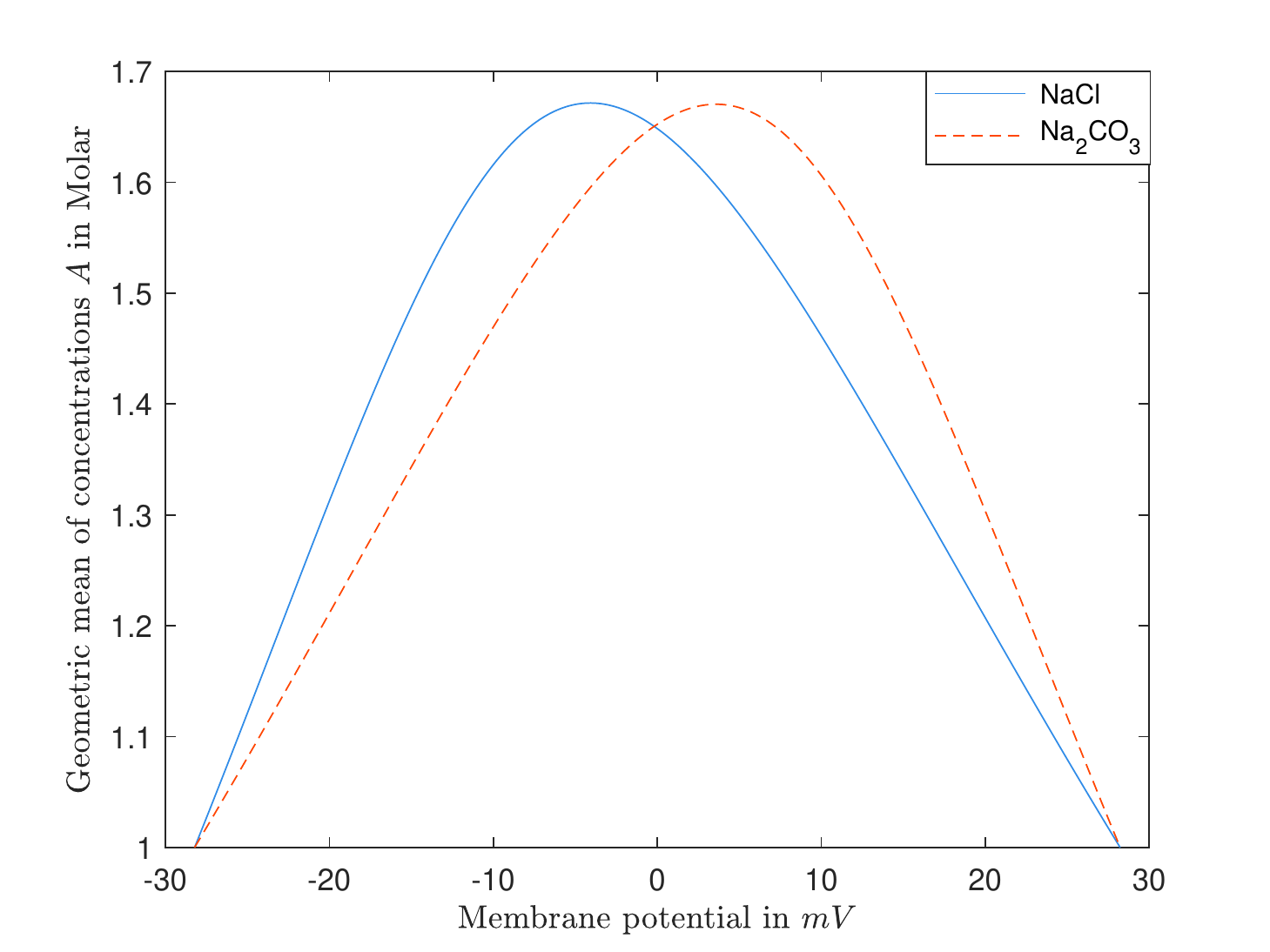} 
\epsfxsize=3.0in \epsfbox{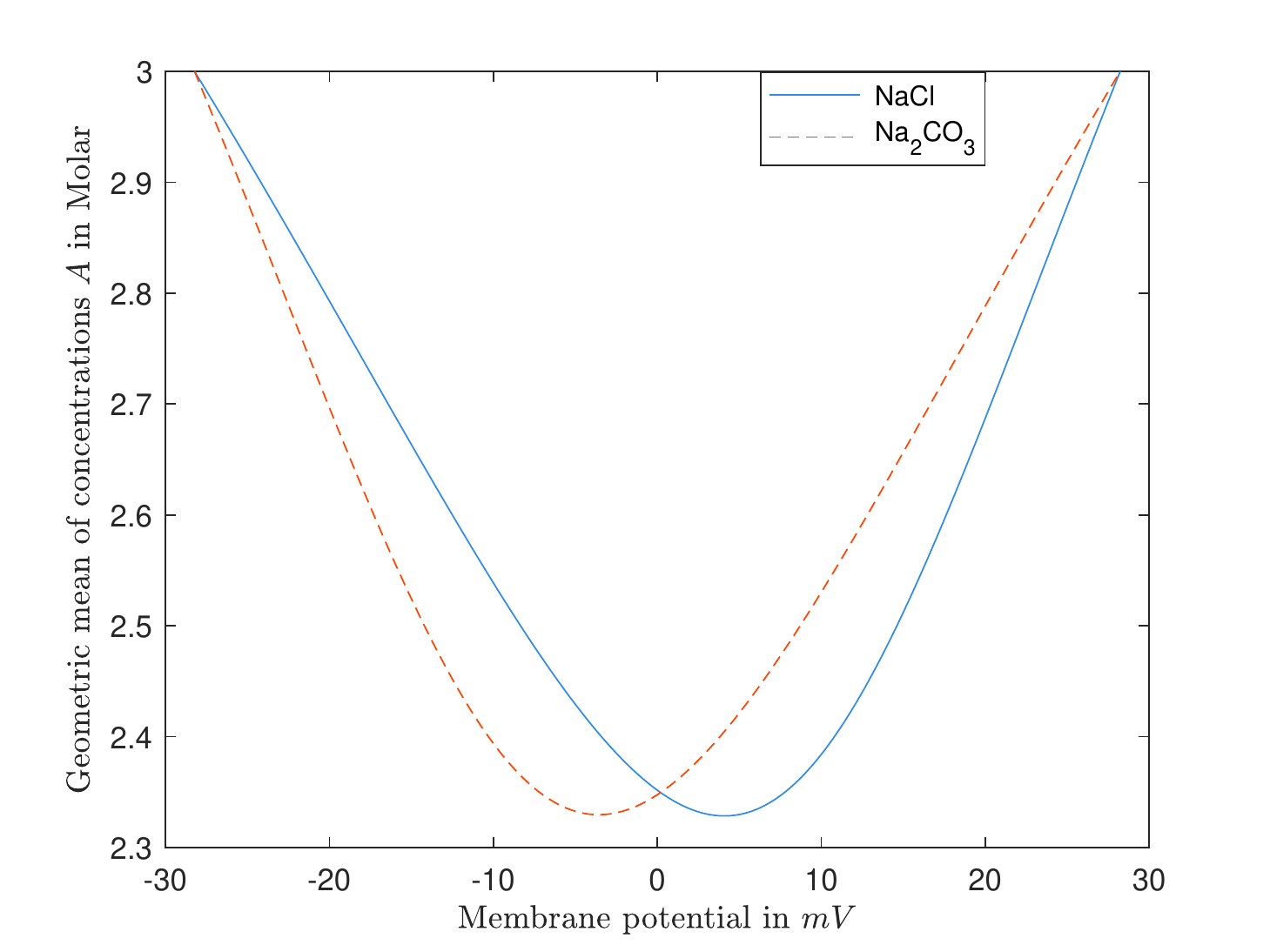} 
}
\caption{\em  The function $A(Q_0,\delta)$  for two pairs of  $(D_1,D_2)$ corresponding to $NaCl$ and $Na_2CO_3$ for various values of $\mathcal{V}$: left panel for $L=0.1 M, R=0.3 M$; right panel for $L=0.3 M, R=0.1 M$.  }
\label{Fig-AvsV}
\end{figure}
 
 \pagebreak  
 
%%%%%%%%%%%%%%%%%%%%%%%%%%%%%%%%%%%%%%%%

\subsection{Results on Reversal Permanent Charge $Q_{rev}$.}\label{revPC}

The spatial distribution of side chains in a specific channel defines the permanent charge of the channel, which forms most of the electrical structure of the channel protein. Thus, reversal potentials should always exist within the ion channels. However, there is a simple necessary condition for the existence of the reversal permanent charge $Q_{rev}$, as one will see in the next Theorem. 
The general result for reversal permanent charge with a given electric potential $V$ is as follows. 

\begin{thm}
 For $n=2$, there exists a unique  reversal permanent charge $Q_{rev}$ if and only if 
\begin{equation}\label{Qrev-NecCond}
\big|V\big| < \Big|\ln\dfrac{l}{r}\Big|.
\end{equation}
\end{thm}

\begin{proof}
Existence of permanent charge $Q_{rev}$ have been proved in \cite{ML19}. The uniqueness is the consequence of Theorems \ref{Thm-A}, \ref{thmA(Q)} and $\partial_{Q_0}G_2$ in \eqref{parG1G2}.
\end{proof}

%%%%%%%%%%%%%%%%%%%%%%%%%%%%%%%%%%%%%%%%55

% Uniqueness: 
%
%\[g_1'(Q_0)= \dfrac{1}{\partial_AG_2}\big(\partial_{Q_0}G_1 \partial_{A}G_2 -\partial_AG_1 \partial_{Q_0}G_2 \big).\]
%%%%%%%%%%%%
%It follows from theorem \ref{thmA(Q)} and lemma \ref{lem-parG1G2}, one gets that $g_1(Q_0)$ is monotone; indeed, if $r < l$ then $g_1'>0$ and if $l<r$ then $g_1'<0$.
%
%\noindent Now, for a given $V_0$, if 
%$ \big(V_0 + \ln\dfrac{l}{r} \big)>0$ and $ \big( -V_0 + \ln\dfrac{l}{r} \big) >0, $
%then $r < l$ and from \eqref{ineqlimg1},
%$$
%\lim_{Q_0\to -\infty}g_1(Q_0) < 0 <  \lim_{Q_0\to +\infty}g_1(Q_0).
%$$
%And if $ \big(V_0 + \ln\dfrac{l}{r} \big)<0$ and $ \big( -V_0 + \ln\dfrac{l}{r} \big) <0, $ then $l < r$ and from \eqref{ineqlimg1},
%$$
%\lim_{Q_0\to +\infty}g_1(Q_0) < 0 <  \lim_{Q_0\to -\infty}g_1(Q_0).
%$$
%Therefore, there exists a (unique) reversal permanent. 

  %%%%%%%%%%%%%%%%%%%%%%%%%%%%%%%%%%%
\begin{thm}\label{Thm-Qrev}
	For any given $(V_0, l, r)$ that satisfies the condition \eqref{Qrev-NecCond} one has,\\  
	(a)  if $l < r$, then $\partial_{V_0}Q_{rev} <0$ and $\lim\limits_{V_0 \to \pm\ln \frac{l}{r} } Q_{rev}(V_0)= \mp\infty$;\\
(b) if $l>r $, then $\partial_{V_0}Q_{rev} >0$ and $\lim\limits_{V_0 \to \pm\ln \frac{l}{r} } Q_{rev}(V_0)= \pm\infty$
\end{thm}
\begin{proof}
	For any given $V_0$ from $G_1\big(V_0, Q_{rev}(V_0), A(Q_{rev}(V_0)) \big)=0$ we have,
	$$
	-1 + \partial_{Q_0}G_1 \partial_{V_0}Q_{rev} + \partial_{A}G_1 \partial_{Q_0}A \partial_{V_0}Q_{rev} =0.
	$$
	It follows from above and $\partial_{Q_0}A = - \dfrac{\partial_{Q_0}G_2}{\partial_{A}G_2}$ that,
	\begin{equation}\label{parV-Q}
	\begin{aligned}
	\partial_{V_0}Q_{rev} =& \dfrac{ \partial_{A}G_2}{\partial_{Q_0}G_1\partial_{A}G_2 -  \partial_{A}G_1 \partial_{Q_0}G_2}.
	\end{aligned}
	\end{equation}
The parts (a) and (b)statements can be verified from \eqref{parV-Q}, and Theorem \ref{Thm-A}.
\end{proof}
The numerical investigations in Figure \ref{Fig-Qrev} admits our results in Theorem \ref{Thm-Qrev}.
It shows the graph of $Q_{rev} $ for values of $V_0$ where $- \ln \frac{l}{r} \leq V_0 \leq \ln \frac{l}{r}$. One can see that when $V_0$ go to $\pm \ln \frac{l}{r}$  then $Q_{rev}$ becomes large unbounded.
%%%%%%%%%%%%%%%%%%%%%%%%%%%%%%%%%%
	\begin{figure}[h]\label{Fig-Qrev}
		\centerline{\epsfxsize=3.0in
			\epsfbox{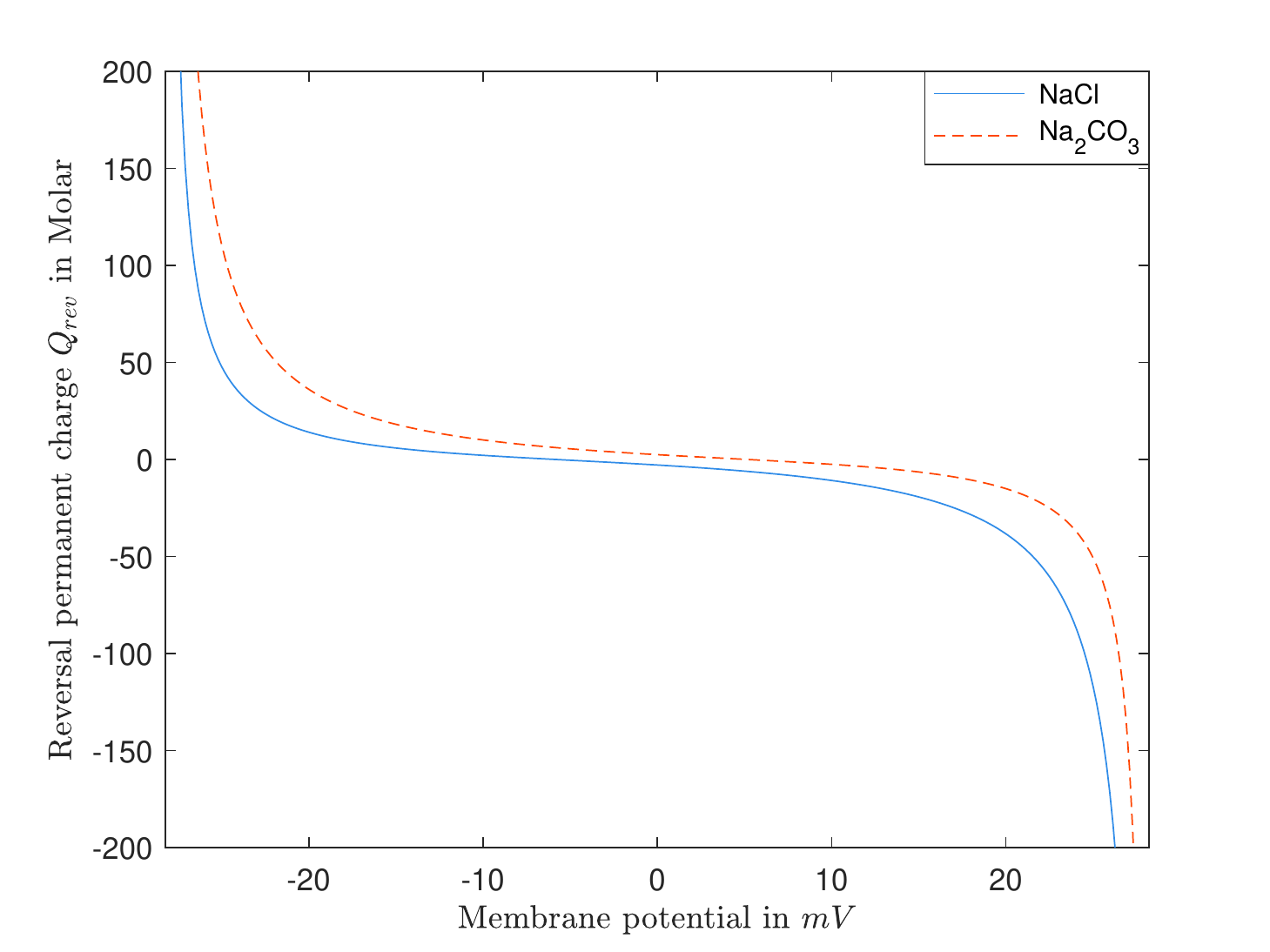} 
			\epsfxsize=3.0in \epsfbox{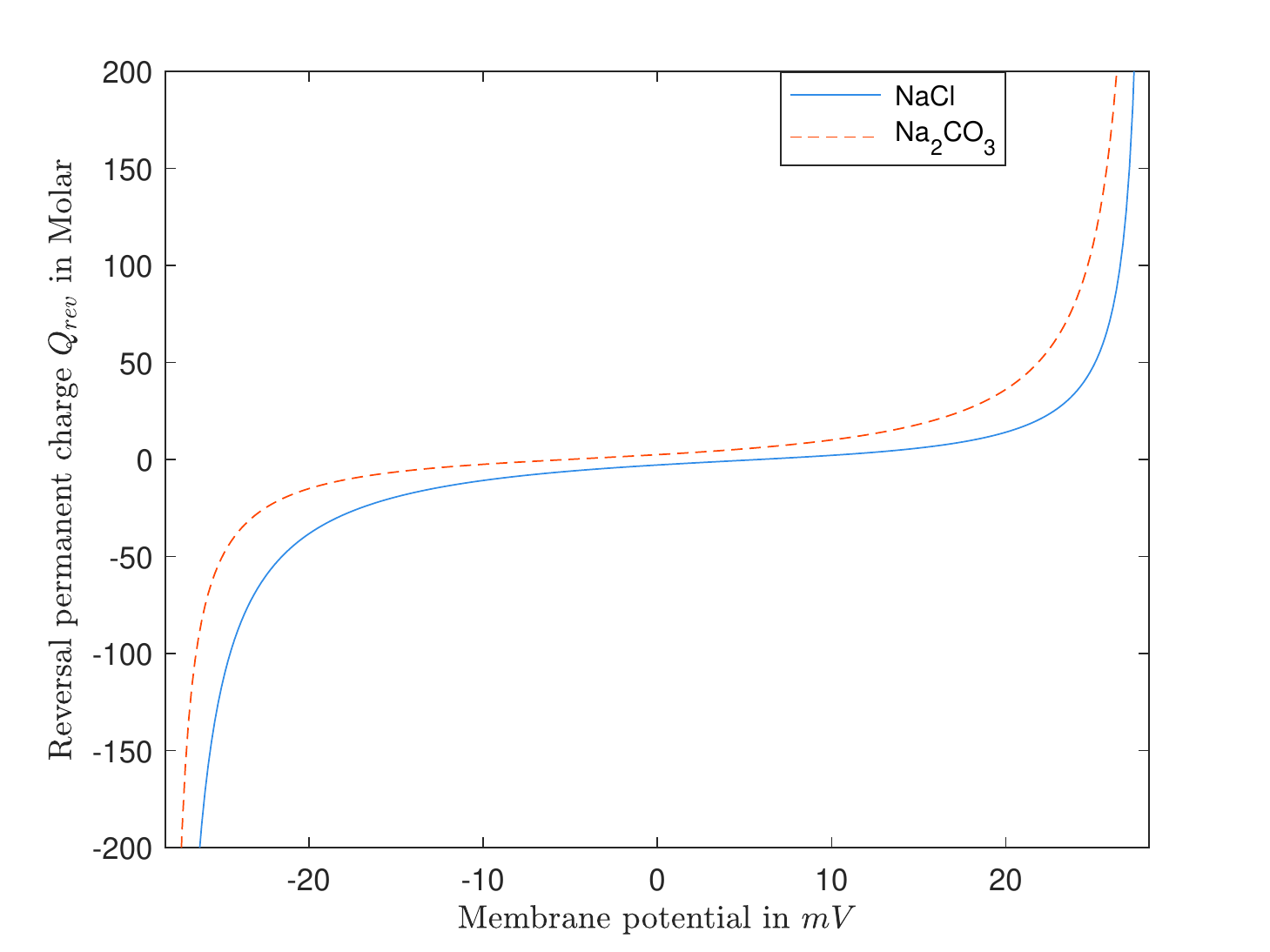}}\vspace*{-.1in}
		\caption{\em The function $Q_{rev}$  for two pairs of  $(D_1,D_2)$ corresponding to $NaCl$ and $Na_2CO_3$: left panel for $L=0.1 M < R=0.3 M$; right panel for $L=0.3 M > R=0.1 M$. }
		\label{Fig-Qrev}
	\end{figure}
%%%%%%%%%%%%%%%%%%%%%%%%%%%

 \section{Concluding Remarks.}\label{ConSec}
 In this manuscript, we work on the classical PNP model allowing unequal  diffusion constants and for a single profile of permanent charges, to study the specific questions about reversal potentials and reversal permanent charges that are among the central issues of biological concerns. 
 
   A crucial assumption is that the dimensionless parameter $\varepsilon$ of the ratio of the Debye length over the distance between the two applied electrodes  is  small. The assumption allows one to treat the PNP system as a singularly perturbed system with $\varepsilon$ as the singular parameter.  
   
Our study relies on a modern general geometric singular perturbation theory and some unique structures of the classical PNP models. Then we obtain a nonlinear matching system of algebraic equations \eqref{Matching} for the zero current condition that includes both the reversal potential and reversal permanent charge topics. We use an intermediate variable introduced in \cite{EL07}, to further reduce the matching system to an effective system of two algebraic equations with two unknowns. Several novel properties of biological significance have resulted from the analysis of these governing equations that some are not intuitive.
In the future, we intend to maintain the study to extend the analysis in this work and explore the problems numerically applying more advanced and complex models. More numerical observations can also be presented to study profiles of relevant physical quantities, e.g., to numerically investigate the behavior of $C_k(X)$, $\Phi(X)$ and  $\mu_k(X)$ throughout the channel. There are many other inspiring projects one can initiate related to this work. Another exciting novel project combines ion size to the problem using hard-sphere electrochemical potentials to analyze ion size's effects.

%%%%%%%%%%%%%%%%%%%%%%%%%%%%%%%%%%%%%%%%%%

%\newpage

\noindent 
{\bf Acknowledgement.}  The author thanks Dr. Weishi Liu for beginning such an exciting topic, for his support and advice.
%%%%%%%%%%%%%%%%%%%%%%%%%%%%%%%%%%%%%%%%%%

%%%%%%%%%%%%%%%%%%%%%%%%%%%%%%%%%%%%%%%%%%%%
 \small
 
  \bibliographystyle{plain}

\end{document}